\documentclass[12pt]{article}
\usepackage{sectsty}
\sectionfont{\large}
\usepackage{amsmath,amsfonts,amsthm,amssymb,amscd,latexsym,multicol}
\usepackage[cp1251]{inputenc}
\usepackage[dvips]{graphicx}
\usepackage[dvips]{color}
\usepackage[all]{xypic}
\newtheorem{theorem}{Theorem}
\newtheorem{theorem*}{Theorem}
\newtheorem{lemma*}{Lemma}

\newtheorem{corollary*}{Corollary}
\newtheorem{corollary}{Corollary}
\newtheorem{remark*}{Remark}

\theoremstyle{remark}

\newtheorem{example*}{Example}
\newtheorem{remark}{Remark}
\begin{document}
\begin{center}
{\Large{Asymptotically Independent Markov Sampling: \\a new MCMC
scheme for Bayesian Inference}}\vspace{2mm}

James L. Beck and Konstantin M. Zuev\footnote{ Both authors contributed equally to this work. Corresponding author's email: zuev@caltech.edu.}

{\footnotesize Computing and Mathematical Sciences, Division of Engineering and Applied
Science, \\California Institute of Technology, USA}

%{\footnotesize October 8, 2011}

%{\footnotesize Draft 8}
\end{center}
\vspace{-6mm} \begin{abstract} In Bayesian statistics, many problems
can be expressed as the evaluation of the expectation of a quantity
of interest with respect to the posterior distribution. Standard
Monte Carlo method is often not applicable because the encountered
posterior distributions cannot be sampled directly. In this case,
the most popular strategies are the importance sampling method,
Markov chain Monte Carlo, and annealing. In this paper, we introduce
a new scheme for Bayesian inference, called Asymptotically
Independent Markov Sampling (AIMS), which is based on the above
methods. We derive important ergodic properties of AIMS. In
particular, it is shown that, under certain conditions, the AIMS
algorithm produces a uniformly ergodic Markov chain. The choice of
the free parameters of the algorithm is discussed and
recommendations are provided for this choice, both theoretically and
heuristically based. The efficiency of AIMS is demonstrated with
three numerical examples, which include
both multi-modal and higher-dimensional target posterior distributions.%\vspace{3mm}
\end{abstract}


\section{Three cornerstones of computational Bayesian inference}
In Bayesian statistics, many problems can be expressed as the
evaluation of the expectation of a quantity of interest with respect
to the posterior distribution. Standard Monte Carlo
simulation \cite{Metropolis_Ulam},
where expectations are estimated by sample averages based on samples
drawn independently from the posterior, is often not applicable because the encountered posterior distributions
are multi-dimensional non-Gaussian distributions that cannot be explicitly normalized.
In this case, the most popular strategies are
importance sampling and Markov chain Monte Carlo methods. We briefly review these two methods first because they
play an important role in the new MCMC method introduced in this paper.

\textit{Importance sampling}: This is nearly as old as the Monte
Carlo method (see, for instance, \cite{Kahn_Marshall}), and works
as follows. Suppose we want to evaluate $\mathbb{E}_\pi[h]$ that is
an expectation of a function of interest
$h:\Theta\rightarrow\mathbb{R}$ under distribution\footnote{Unless
otherwise stated, all probability distributions are assumed
%to be absolutely continuous with respect to Lebesgue measure,
to have densities with respect to Lebesgue measure,
$\pi(d\theta)=\pi(\theta)d\theta$. For simplicity, the same symbol
will be used to denote both the distribution and its density, and we write $\theta\sim\pi(\cdot)$ to denote that $\theta$ is distributed
according to $\pi(\cdot)$.}
$\pi(\cdot)$ defined on a parameter space
$\Theta\subseteq\mathbb{R}^d$,
\begin{equation}\label{E[h]}
    \mathbb{E}_\pi[h]=\int_\Theta h(\theta)\pi(\theta)d\theta.
\end{equation}
Suppose also that we are not able to sample directly from $\pi(\cdot)$,
although we can compute $\pi(\theta)$ for any $\theta\in\Theta$ to within a proportionality constant.
Instead, we sample from some other distribution
$q(\cdot)$ on $\Theta$ which is readily computable for any
$\theta\in\Theta$. Let $\theta^{(1)},\ldots,\theta^{(N)}$ be $N$
i.i.d. samples from $q(\cdot)$, and
$w^{(i)}=\pi(\theta^{(i)})/q(\theta^{(i)})$ denote the
\textit{importance weight} of the $i^{\mathrm{th}}$ sample, then we
can estimate $\mathbb{E}_\pi[h]$ by
\begin{equation}\label{IS}
    \hat{h}_N=\frac{\sum_{i=1}^Nw^{(i)}h(\theta^{(i)})}{\sum_{i=1}^Nw^{(i)}}.
\end{equation}
 The estimator $\hat{h}_N$ converges almost surely as
$N\rightarrow\infty$ to $\mathbb{E}_\pi[h]$ by the Strong Law of
Large Numbers for any choice of distribution $q(\cdot)$, provided
$\mathrm{supp}(\pi)\subseteq\mathrm{supp}(q)$. Note that the latter
condition automatically holds in Bayesian updating using data $\mathcal{D}$ where
$q(\theta)=\pi_0(\theta)$ is the prior density and
$\pi(\theta)\propto \pi_0(\theta)L(\theta)$ is the posterior $p(\theta|\mathcal{D})$, where
$L$ stands for the likelihood function $p(\mathcal{D}|\theta)$.

The estimator $\hat{h}_N$ in (\ref{IS}) generally
has a smaller mean square error than a more
straightforward unbiased importance sampling estimator:
\begin{equation}\label{IS_2}
    \hat{h}'_N=\frac{1}{N}\sum_{i=1}^Nw^{(i)}h(x^{(i)}).
\end{equation}
This is
especially clear when $h$ is nearly a constant: if $h\approx c$,
then $\hat{h}_N\approx c$, while $\hat{h}'_N$ has a larger
variation. Although $\hat{h}_N$ is biased for any finite $N$, the bias can be made small by taking sufficiently large $N$,
and the improvement in variance makes it a preferred alternative to
$\hat{h}'_N$ \cite{Liu,RobCas}. Another major advantage of using
$\hat{h}_N$ instead of $\hat{h}'_N$, which is especially
important for Bayesian applications, is that in using the former we need to know  $\pi(\theta)$ only up to a multiplicative normalizing constant;
whereas in the latter, this constant must be known exactly.

The accuracy of $\hat{h}_N$  depends critically on the choice of the
\textit{importance sampling distribution} (ISD) $q(\cdot)$, which is also called
the \textit{instrumental} or \textit{trial} distribution. If
$q(\cdot)$ is chosen carelessly such that the the importance weights
$w^{(i)}$ have a large variation, then $\hat{h}_N$ is essentially
based only on the few samples $\theta^{(i)}$ with the largest
weights, yielding generally a very poor estimate. Hence, for
importance sampling to work efficiently, $q(\cdot)$ must be a good
approximation of $\pi(\cdot)$ --- ``the importance sampling density
should mimic the posterior density'' \cite{Geweke} --- so that the
variance $\mathrm{var}_q[w]$ is not large. Since usually the prior and
posterior are quite different, it is, therefore, highly inefficient
to use the prior as the importance sampling distribution. When $\Theta$
is high-dimensional, and $\pi(\cdot)$ is complex, finding a good
importance sampling distribution can be very challenging, limiting
the applicability of the method \cite{AuBeck2}.

For the estimator
$\hat{h}'_N$ in (\ref{IS_2}), it is not difficult to show that the optimal
importance sampling density, i.e., $q^*(\cdot)$ that minimizes the
variance of $\hat{h}'_N$, is $q^*(\theta)\propto
|h(\theta)|\pi(\theta)$. This result is sometimes attributed to
Rubinstein \cite{Rubinstein}, although it was proved earlier
 by Kahn and Marshall \cite{Kahn_Marshall}. It is not
true, however, that $q^*(\cdot)$ is optimal for the estimator
$\hat{h}_N$. Note also that this optimality result is not useful in practice, since when $h(\theta)\geq0$, the required
normalizing constant of $q^*(\cdot)$ is $\int_\Theta
h(\theta)\pi(\theta)d\theta$, the integral of interest.

\textit{MCMC Sampling}: Instead of generating independent samples from an ISD,
we could generate dependent samples by simulating a
Markov chain whose state distribution converges to the posterior distribution $\pi(\cdot)$ as its stationary distribution.
\textit{Markov chain Monte Carlo} sampling (MCMC) originated in
statistical physics, and now is widely used in solving statistical
problems \cite{NealMCMC,Gilks_Richardson_Spiegelhalter,Liu,RobCas}.

The Metropolis-Hastings algorithm \cite{Metropolis,Hastings}, the
most popular MCMC technique, works as follows. Let $q(\cdot|\theta)$
be a distribution on $\Theta$, which may or may not depend on
$\theta\in\Theta$. Assume that $q(\cdot|\theta)$ is easy to sample
from and it is either computable (up to a multiplicative constant)
or symmetric, i.e. $q(\xi|\theta)=q(\theta|\xi)$. The sampling distribution
$q(\cdot|\theta)$ is called the \textit{proposal distribution}.
Starting from essentially any $\theta^{(1)}\in\mathrm{supp}(\pi)$, the
Metropolis-Hastings algorithm proceeds by iterating the following
two steps. First, generate a \textit{candidate} state $\xi$ from the
proposal density $q(\cdot|\theta^{(n)})$. Second, either accept $\xi$ as
the next state of the Markov chain, $\theta^{(n+1)}=\xi$, with
probability
$\alpha(\xi|\theta^{(n)})=\min\left\{1,\frac{\pi(\xi)q(\theta^{(n)}|\xi)}{\pi(\theta^{(n)})q(\xi|\theta^{(n)})}\right\}$;
or reject $\xi$ and set $\theta^{(n+1)}=\theta^{(n)}$ with the
remaining probability $1-\alpha(\xi|\theta^{(n)})$. It can be shown
(see, for example, \cite{RobCas}), that under fairly weak
conditions, $\pi(\cdot)$ is the stationary distribution of the
Markov chain $\theta^{(1)}, \theta^{(2)},\ldots$ and
\begin{equation}\label{convergence}
    \lim_{N\rightarrow\infty}\frac{1}{N}\sum_{i=1}^N
    h(\theta^{(i)})=\int_\Theta h(\theta)\pi(\theta)d\theta.
\end{equation}
Since the chain needs some time (so called ``burn-in'' period) to
converge to stationarity, in practice, an initial portion of, say, $N_0$ states
is usually discarded and
\begin{equation}\label{MCMCestimator}
   \tilde{h}_N=\frac{1}{N-N_0}\sum_{i=N_0+1}^Nh(\theta^{(i)})
\end{equation}
is used as an estimator for $\mathbb{E}_{\pi}[h]$.

The two main special cases of the Metropolis-Hastings algorithm are
Independent Metropolis-Hastings (IMH), where the proposal
distribution $q(\xi|\theta)=q_g(\xi)$ is independent of $\theta$ (so $q_g$ is a \textit{global proposal}),
and Random Walk Metropolis-Hastings (RWMH), where the proposal
distribution is of the form $q(\xi|\theta)=q_l(\xi-\theta)$, i.e. a
candidate state is proposed as $\xi=\theta^{(n)}+\epsilon_n$, where
$\epsilon_n\sim q_l(\cdot)$ is a random perturbation (so $q_l$ is a \textit{local proposal}). In both cases,
the choice of the proposal distribution strongly affects the
efficiency of the algorithms. For IMH to work well, as with
importance sampling, the proposal distribution must be a good
approximation of the \textit{target distribution} $\pi(\cdot)$, otherwise a large fraction of the candidate samples will be
rejected and the Markov chain will be too slow in covering the important regions for $\pi(\cdot)$.
When, however, it is possible to
find a proposal $q_g(\cdot)$, such that $q_g(\cdot)\approx
\pi(\cdot)$, IMH should always be preferred to RWMH because of better efficiency, i.e. better
approximations of $\mathbb{E}_{\pi}[h]$ for a given number of samples $N$. Unfortunately,
such a proposal is difficult to construct in the context of Bayesian
inference where the posterior $\pi(\cdot)$ is often complex and
high-dimensional. This limits the applicability of IMH.

Since the random walk proposal $q_l(\cdot)$ is local, it is less
sensitive to the target distribution. That is why, in practice, RWMH
is more robust and used more frequently than IMH. Nonetheless, there
are settings where RWMH also does not work well because of the
complexity of the posterior distribution. Although
(\ref{convergence}) is true in theory, a potential problem with RWMH
(and, in fact, with any MCMC algorithm) is that the generated samples
$\theta^{(1)},\ldots, \theta^{(N)}$ often consist of highly
correlated samples. Therefore, the estimator $\tilde{h}_N$ in
(\ref{MCMCestimator}) obtained from these samples tends to have a
large variance for a modest amount of samples. This is especially true when the posterior
distribution contains several widely-separated modes: a chain will
move between modes only rarely and it will take a long time before
it reaches stationarity. If this is the case, an estimate produced
by $\tilde{h}_N$ will be very inaccurate. At first glance, it seems
natural to generate several independent Markov chains, starting from
different random seeds, and hope that different chains will get
trapped by different modes. However, multiple runs will not in
general generate a sample in which each mode is correctly
represented, since the probability of a chain reaching a mode
depends more on the mode's ``basin of attraction'' than on the
probability concentrated in the mode
\cite{Neal_tempered_transitions}.

\textit{Annealing:} The concept of \textit{annealing} (or \textit{tempering}), which involves moving from an easy-to-sample
distribution to the target distribution via a sequence of
intermediate distributions, is one of the most effective methods of
handling multiple isolated modes. Together with importance sampling and MCMC,
annealing constitutes the third cornerstone of computational
Bayesian inference.

The idea of using the RWMH algorithm in conjunction with annealing
was introduced independently in \cite{Kirkpatrick} and \cite{Cerny}
for solving difficult optimization problems. The resulting
algorithm, called \textit{Simulated Annealing}, works as follows. Suppose we
want to find the global minimum of a function of interest $h: \Theta
\rightarrow \mathbb{R}$. This is equivalent to finding the global
maximum of $f_T(\theta)=\exp(-h(\theta)/T)$ for any given $T>0$. By
analogy with the Gibbs distribution in statistical mechanics,
$T$ is called the \textit{temperature parameter}. Let $T_0>T_1>\ldots$ be a sequence of
monotonically decreasing temperatures, in which $T_0$ is large
enough so that the probability distribution $\pi_0(\theta)\propto
f_{T_0}(\theta)$ is close to uniform, and
$\lim_{j\rightarrow\infty}T_j=0$. At each temperature $T_j$, the
Simulated Annealing method generates a Markov chain with
$\pi_j(\theta)\propto \exp(-h(\theta)/T_j)$ as its stationary
distribution. The final state of the Markov chain at simulation
level $j$ is used as the initial state for the chain at level $j+1$.
The key observation is that for any function $h$ such that
$\int_\Theta \exp(-h(\theta)/T)d\theta<\infty$ for all $T>0$,
distribution $\pi_j(\cdot)$, as $j$ increases, puts more and more of
its probability mass (converging to $1$) into a neighborhood of the
global minimum of $h$. Therefore, a sample drawn from $\pi_j(\cdot)$
would almost surely be in a vicinity of the global minimum of $h$
when $T_j$ is close to zero.

The success of Simulated Annealing in finding the global minimum
crucially depends on the schedule of temperatures used in the
simulation. It was proved in \cite{Geman} that if a logarithmic
schedule $T_j=T_0/\log(j+1)$ is used, then, under certain
conditions, there exists a value for $T_0$ such that use of this
schedule guarantees that the global minimum of $h$ will be reached
almost surely. In practice, however, such a slow annealing schedule is not computationally efficient.
It is more common to use either a geometric schedule, $T_{j+1}=\gamma T_j$ with $0<\gamma<1$, or
some adaptive schedule, which defines the temperature for the next
annealing level based on characteristics of the samples observed at
earlier levels. For examples of adaptive annealing schedules, see,
for instance, \cite{NealMCMC}.

In Bayesian inference problems, the idea of annealing is typically
employed in the following way. First, we construct (in advance or
adaptively) a sequence of distributions $\pi_0(\cdot),\ldots,
\pi_m(\cdot)$ interpolating between the prior distribution
$\pi_0(\cdot)$ and the posterior distribution $\pi(\cdot)\equiv
\pi_m(\cdot)$. Next, we generate i.i.d. samples
$\theta_0^{(1)},\ldots, \theta_0^{(N)}$ from the prior, which is
assumed to be readily sampled. Then, at each annealing level $j$,
using some MCMC algorithm and samples $\theta_{j-1}^{(1)},\ldots,
\theta_{j-1}^{(N)}$ from the previous level $j-1$, we generate
samples $\theta_{j}^{(1)},\ldots, \theta_{j}^{(N)}$ which
are approximately distributed according to $\pi_j(\cdot)$. We
proceed sequentially in this way, until the posterior distribution has
been sampled. The rationale behind this strategy is that sampling
from the multi-modal and, perhaps, high-dimensional posterior in
such a way is likely to be more efficient than a straightforward
MCMC sampling of the posterior.

The problem of sampling a complex distribution is encountered in
statistical mechanics, computational Bayesian inference, scientific
computing, machine learning, and other fields. As a result, many
different efficient algorithms have been recently developed, e.g.
the method of Simulated Tempering \cite{Marinari,Geyer}, the
Tempered Transition method \cite{Neal_tempered_transitions},
Annealed Importance Sampling \cite{Neal_AIS}, the Adaptive Metropolis-Hastings algorithm \cite{BeckAu}, Transitional Markov
Chain Monte Carlo method \cite{Ching}, to name a few.

In this paper we introduce a new MCMC scheme for Bayesian inference,
called \textit{Asymptotically Independent Markov Sampling} (AIMS),
which combines the three approaches described above --- importance
sampling, MCMC, and annealing --- in the following way. Importance
sampling with $\pi_{j-1}(\cdot)$ as the ISD is used
for a construction of an approximation $\hat{\pi}_{j}^N(\cdot)$ of
$\pi_{j}(\cdot)$, which is based on samples
$\theta_{j-1}^{(1)},\ldots,\theta_{j-1}^{(N)}\sim \pi_{j-1}(\cdot)$.
This approximation is then employed as the independent (global)
proposal distribution for sampling from $\pi_{j}(\cdot)$ by the IMH
algorithm. Intermediate distributions $\pi_0(\cdot),\ldots,
\pi_m(\cdot)$ interpolating between prior and posterior are
constructed adaptively, using the essential sample size (ESS) to
measure how much $\pi_{j-1}(\cdot)$ differs from $\pi_{j}(\cdot)$. When
the number of samples $N\rightarrow\infty$, the
approximation $\hat{\pi}_{j}^N(\cdot)$ converges to
$\pi_{j}(\cdot)$, providing the optimal proposal distribution. In
other words, when $N\rightarrow\infty$, the corresponding MCMC
sampler produces independent samples, hence the name of the
algorithm.

\begin{remark} The term ``Markov sampling'' has several different
meanings. In this paper it is used as synonymous to ``MCMC sampling''.
\end{remark}

In this introductory section, we have described all the main
ingredients that we will need in the subsequent sections. The rest of
the paper is organized as follows. In Section \ref{sec2}, the AIMS algorithm is described. The ergodic properties of AIMS are derived in Section \ref{Ergodic Properties}.
The efficiency of AIMS is illustrated in Section \ref{examples-section} with three numerical examples that include both multi-modal and
high-dimensional posterior distributions. Concluding remarks are made in Section
\ref{finish}.

\section{Asymptotically Independent Markov Sampling}\label{sec2}

Let $\pi_0(\cdot)$ and $\pi(\cdot)$ be the prior and the posterior
distributions defined on a parameter space $\Theta$, respectively,
so that, according to Bayes' Theorem, $\pi(\theta)\propto
\pi_0(\theta)L(\theta)$, where $L$ denotes the likelihood function for data $\mathcal{D}$.
Our ultimate goal is to draw samples that are distributed according to $\pi(\cdot)$.

In Asymptotically Independent Markov Sampling (AIMS), we
sequentially generate samples from intermediate distributions
$\pi_0(\cdot),\ldots,\pi_m(\cdot)$ interpolating between the prior
$\pi_0(\cdot)$ and the posterior $\pi(\cdot)\equiv\pi_m(\cdot)$. The
sequence of distributions could be specially constructed for a
given problem but the following scheme \cite{Neal_AIS,Ching}
generally yields good efficiency:
\begin{equation}\label{InterPDFs}
    \pi_j(\theta)\propto
    \pi_0(\theta)L(\theta)^{\beta_j},
\end{equation}
where $0=\beta_0<\beta_1<\ldots<\beta_m=1$. We will refer to $j$ and
$\beta_j$ as the \textit{annealing level} and the \textit{annealing
parameter} at level $j$, respectively. In the next subsection, we assume that $\beta_j$ is given and
therefore the intermediate distribution $\pi_j(\cdot)$ is also known.
In Subsection \ref{Implementation}, we describe how to choose the
annealing parameters adaptively.

\subsection{AIMS at annealing level $j$}\label{level j}

Our first goal is to describe how AIMS generates sample
$\theta_{j}^{(1)},\ldots, \theta_{j}^{(N_j)}$ from $\pi_j(\cdot)$
based on the sample $\theta_{j-1}^{(1)},\ldots,
\theta_{j-1}^{(N_{j-1})}\sim\pi_{j-1}(\cdot)$ obtained at the
previous annealing level. We start with an informal motivating
discussion that leads to the simulation algorithm. In Section
\ref{Ergodic Properties}, we rigorously prove that the corresponding
algorithm indeed generates samples which are asymptotically
distributed according to $\pi_j(\cdot)$, as the sample size
$N_j\rightarrow\infty$. Moreover, the larger $N_{j-1}$, the less
correlated generated samples $\theta_{j}^{(1)},\ldots,
\theta_{j}^{(N_j)}$ are --- a very desirable, yet rarely affordable,
property for any MCMC algorithm.

Let $K_j(\cdot|\cdot)$ be any transition kernel %$K_j:
%\mathcal{B}(\Theta)\times\Theta\rightarrow \mathbb{R}$, where
%$\mathcal{B}(\Theta)$ denotes the Borel $\sigma$-field on $\Theta$,
such that $\pi_j(\cdot)$ is a stationary distribution with respect
to $K_j(\cdot|\cdot)$. By definition, this means that
\begin{equation}\label{stationarity}
    \pi_j(\theta)d\theta=\int_\Theta
    K_j(d\theta|\xi)\pi_j(\xi)d\xi
\end{equation}
Applying importance sampling with the sampling density
$\pi_{j-1}(\cdot)$ to integral (\ref{stationarity}), we have:
\begin{equation}\label{ImportanceSamplingApproximation}
\begin{split}
    \pi_j(\theta)d\theta&=\int_\Theta K_j(d\theta|\xi)
    \frac{\pi_j(\xi)}{\pi_{j-1}(\xi)}\pi_{j-1}(\xi)d\xi\\
    &\approx\sum_{i=1}^{N_{j-1}}
    K_j(d\theta|\theta_{j-1}^{(i)})\bar{w}^{(i)}_{j-1}
    \overset{\underset{\mathrm{def}}{}}{=} \hat{\pi}_{j}^{N_{j-1}}(d\theta),
\end{split}
\end{equation}
where $\hat{\pi}_{j}^{N_{j-1}}(\cdot)$ will be used as the \textit{global proposal} distribution in the Independent Metropolis-Hastings algorithm, and
\begin{equation}\label{weights}
    w^{(i)}_{j-1}=\frac{\pi_j(\theta_{j-1}^{(i)})}{\pi_{j-1}(\theta_{j-1}^{(i)})}\propto L(\theta_{j-1}^{(i)})^{\beta_j-\beta_{j-1}}\hspace{3mm}
    \mbox{ and }
    \hspace{3mm} \bar{w}^{(i)}_{j-1}=\frac{w^{(i)}_{j-1}}{\sum_{k=1}^{N_{j-1}}w^{(k)}_{j-1}} %i=1,\ldots, N,
\end{equation}
are the importance weights and normalized importance weights,
respectively. Note that to calculate $\bar{w}^{(i)}_{j-1}$, we do
not need to know the normalizing constants of $\pi_{j-1}(\cdot)$ and
$\pi_{j}(\cdot)$. If adjacent intermediate distributions
$\pi_{j-1}(\cdot)$ and $\pi_{j}(\cdot)$ are sufficiently close (in
other words, if $\Delta\beta_j=\beta_j-\beta_{j-1}$ is small enough),
then the importance weights (\ref{weights}) will not vary wildly, and,
therefore, we can expect that, for reasonably large $N_{j-1}$,
approximation $(\ref{ImportanceSamplingApproximation})$ is accurate.

\begin{remark} In \cite{Cheung}, the stationary condition (\ref{stationarity}) was used
for an analytical approximation of the target PDF to evaluate the evidence (marginal likelihood) for a model.
\end{remark}
\begin{remark} Note that for any finite $N_{j-1}$, distribution
$\hat{\pi}_{j}^{N_{j-1}}(\cdot)$ will usually have both continuous and discrete
parts. This follows from the fact that the transition kernel in
Markov chain simulation usually has the following form:
$K(d\theta|\xi)=k(\theta|\xi)d\theta+r(\xi)\delta_{\xi}(d\theta)$,
where $k(\cdot|\cdot)$ describes the continuous part of the
transition kernel, $\delta_{\xi}(\cdot)$ denotes the Dirac mass at
$\xi$, and $r(\xi)=1-\int_\Theta k(\theta|\xi)d\theta$. This is the form, for example, for the Metropolis-Hastings algorithm.
Therefore, (\ref{ImportanceSamplingApproximation}) must be
understood as the approximate equality of distributions, not
densities. In other words, (\ref{ImportanceSamplingApproximation})
means that
$\mathbb{E}_{\hat{\pi}_{j}^{N_{j-1}}}[h]\approx\mathbb{E}_{\pi_j}[h]$
and
$\mathbb{E}_{\hat{\pi}_{j}^{N_{j-1}}}[h]\rightarrow\mathbb{E}_{\pi_j}[h]$,
when $N_{j-1}\rightarrow\infty$, for all integrable functions $h$.
See also Example 2.1 below.
\end{remark}

From now on, we consider a special case where ${K}_j(\cdot|\cdot)$
is the random walk Metropolis-Hastings (RWMH) transition kernel. In
this case, it can be written as follows:
\begin{equation}\label{RWMHkernel}
    K_j(d\theta|\xi)=q_j(\theta|\xi)\min\left\{1,\frac{\pi_j(\theta)}{\pi_j(\xi)}\right\}d\theta
    + (1-a_j(\xi))\delta_\xi(d\theta),
\end{equation}
where $q_j(\cdot|\xi)$ is a symmetric \textit{local} proposal density,
and $a_j(\xi)$ is the probability of having a proper
transition $\xi$ to $\Theta\setminus\{\xi\}$:
\begin{equation}\label{acceptProb}
    a_j(\xi)=\int_\Theta
    q_j(\theta|\xi)\min\left\{1,\frac{\pi_j(\theta)}{\pi_j(\xi)}\right\}
    d\theta
\end{equation}

%\begin{example*}\label{ex1}
\textit{Example 2.1.}
As a simple illustration of (\ref{ImportanceSamplingApproximation}),
consider the case when $\pi_j(\cdot)=\mathcal{N}(\cdot|0,1)$,
$\pi_{j-1}(\cdot)=\mathcal{N}(\cdot|0,2)$, and
$q_j(\cdot|\xi)=\mathcal{N}(\cdot|\xi,1/2)$, where
$\mathcal{N}(\cdot|\mu,\sigma^2)$ denotes the Gaussian density with
mean $\mu$ and variance $\sigma^2$. The approximation
$\hat{\pi}_{j}^{N_{j-1}}(\cdot)$ based on the samples
$\theta_{j-1}^{(1)},\ldots,
\theta_{j-1}^{(N_{j-1})}\sim\mathcal{N}(\cdot|0,2)$ is shown in the
top panels of Figure~\ref{example1}, for $N_{j-1}=5$ and
$N_{j-1}=50$. Suppose that $h_1(\theta)=\theta$ and
$h_2(\theta)=\theta^2$ are the functions of interest. Then
$\mathbb{E}_{\pi_j}[h_1]=0$ and $\mathbb{E}_{\pi_j}[h_2]=1$. The
convergence of
$h^*_1(N_{j-1})=\mathbb{E}_{\hat{\pi}_{j}^{N_{j-1}}}[h_1]$ and
$h^*_2(N_{j-1})=\mathbb{E}_{\hat{\pi}_{j}^{N_{j-1}}}[h_2]$ is shown in
the bottom panel of Figure~\ref{example1}.
%\end{example*}

For sampling from $\pi_j(\cdot)$, we will use the
Independent Metropolis-Hastings algorithm (IMH) with the \textit{global}
proposal distribution $\hat{\pi}_j^{N_{j-1}}(\cdot)$. To accomplish
this, we have to be able to calculate the ratio
$\hat{\pi}_j^{N_{j-1}}(\theta)/\hat{\pi}_j^{N_{j-1}}(\xi)$ for any
$\theta,\xi\in \Theta$ as a part of the expression for the
acceptance probability
$\alpha_j(\xi|\theta)=\min\left\{1,\frac{\pi_j(\xi)\hat{\pi}_j^{N_{j-1}}(\theta)}{\pi_j(\theta)\hat{\pi}_j^{N_{j-1}}(\xi)}\right\}$.
However, as it has been already mentioned, the distribution
$\hat{\pi}_j^{N_{j-1}}(\cdot)$ does not have a density since it has
both continuous and discrete components, and, therefore, the ratio
$\hat{\pi}_j^{N_{j-1}}(\theta)/\hat{\pi}_j^{N_{j-1}}(\xi)$ makes no
sense. To overcome this ``lack-of-continuity problem'', taking into
account (\ref{ImportanceSamplingApproximation}) and
(\ref{RWMHkernel}), let us \textit{formally} define the global proposal distribution over $\Theta$ as:
\begin{equation}\label{formaldefinition}
\hat{\pi}_{j}^{N_{j-1}}(\theta)\overset{\underset{\mathrm{def}}{}}{=}\sum_{i=1}^{N_{j-1}}\bar{w}^{(i)}_{j-1}q_j(\theta|\theta^{(i)}_{j-1})\min\left\{1,\frac{\pi_j(\theta)}{\pi_j(\theta^{(i)}_{j-1})}\right\},
\end{equation}
if
$\theta\notin\left\{\theta_{j-1}^{(1)},\ldots,\theta_{j-1}^{(N_{j-1})}\right\}$,
and
\begin{equation}\label{formaldefinition2}
\hat{\pi}_{j}^{N_{j-1}}(\theta_{j-1}^{(k)})\overset{\underset{\mathrm{def}}{}}{=}\infty%n_k\bar{w}^{(k)}_{j-1}(1-a_j(\theta_{j-1}^{(k)}))\delta(0),
\end{equation}
%where $n_k$ is the number of copies of $\theta_{j-1}^{(k)}$ in the
%sample $\theta_{j-1}^{(1)},\ldots,\theta_{j-1}^{(N_{j-1})}$.
Note that $\hat{\pi}_{j}^{N_{j-1}}(\cdot)$ is a distribution on $\Theta$, but it does not have a density.
However, $\hat{\pi}_{j}^{N_{j-1}}(\cdot)$ induces another distribution on
$\Theta\setminus\left\{\theta_{j-1}^{(1)},\ldots,\theta_{j-1}^{(N_{j-1})}\right\}$ which does have a density,
given by the r.h.s. of (\ref{formaldefinition}). This motivates (\ref{formaldefinition}).

Now, using (\ref{formaldefinition}) and (\ref{formaldefinition2}), we can
calculate the ratio
$\hat{\pi}_j^{N_{j-1}}(\theta)/\hat{\pi}_j^{N_{j-1}}(\xi)$ as
follows:
\begin{enumerate}
  \item[I.] If
  $\theta,\xi\notin\left\{\theta_{j-1}^{(1)},\ldots,\theta_{j-1}^{(N_{j-1})}\right\}$,
  then
  \begin{equation}\label{case1}
    \frac{\hat{\pi}_j^{N_{j-1}}(\theta)}{\hat{\pi}_j^{N_{j-1}}(\xi)}=
    \frac{\sum_{i=1}^{N_{j-1}}\bar{w}^{(i)}_{j-1}q_j(\theta|\theta^{(i)}_{j-1})\min\left\{1,\frac{\pi_j(\theta)}{\pi_j(\theta^{(i)}_{j-1})}\right\}}
    {\sum_{i=1}^{N_{j-1}}\bar{w}^{(i)}_{j-1}q_j(\xi|\theta^{(i)}_{j-1})\min\left\{1,\frac{\pi_j(\xi)}{\pi_j(\theta^{(i)}_{j-1})}\right\}}
  \end{equation}
  \item[II.] If
  $\theta\notin\left\{\theta_{j-1}^{(1)},\ldots,\theta_{j-1}^{(N_{j-1})}\right\}$
  and $\xi=\theta_{j-1}^{(k)}$, then
  \begin{equation}\label{case2}
\frac{\hat{\pi}_j^{N_{j-1}}(\theta)}{\hat{\pi}_j^{N_{j-1}}(\xi)}=0 \hspace{3mm} \mbox{and} \hspace{3mm} \alpha_j(\xi|\theta)=0
  \end{equation}
  \item[III.] If $\theta=\theta_{j-1}^{(k)}$ and
  $\xi\notin\left\{\theta_{j-1}^{(1)},\ldots,\theta_{j-1}^{(N_{j-1})}\right\}$,
  then
  \begin{equation}\label{case3}
\frac{\hat{\pi}_j^{N_{j-1}}(\theta)}{\hat{\pi}_j^{N_{j-1}}(\xi)}=\infty \hspace{3mm} \mbox{and} \hspace{3mm} \alpha_j(\xi|\theta)=1
  \end{equation}
  \item[IV.] If $\theta=\theta_{j-1}^{(k)}$ and $\xi=\theta_{j-1}^{(l)}$, then $\frac{\hat{\pi}_j^{N_{j-1}}(\theta)}{\hat{\pi}_j^{N_{j-1}}(\xi)}$ is not defined.
%\begin{equation}\label{case4}
%\frac{\hat{\pi}_j^{N_{j-1}}(\theta)}{\hat{\pi}_j^{N_{j-1}}(\xi)}=
%\frac{n_k\bar{w}^{(k)}_{j-1}(1-a_j(\theta_{j-1}^{(k)}))}{n_l\bar{w}^{(l)}_{j-1}(1-a_j(\theta_{j-1}^{(l)}))}
%\end{equation}
\end{enumerate}

Notice that in the first three cases the ratio
$\hat{\pi}_j^{N_{j-1}}(\theta)/\hat{\pi}_j^{N_{j-1}}(\xi)$ is
readily computable, while in Case IV, it is not even defined. %Expression
%(\ref{case4}) contains $a_j(\theta_{j-1}^{(k)})$ and
%$a_j(\theta_{j-1}^{(l)})$, which are given by integral
%(\ref{acceptProb}), and evaluation of this integral is a non-trivial
%computational task.
Therefore, it is very desirable to avoid Case
IV. The key observation that allows us to do this is the following:
suppose that the initial state $\theta^{(1)}_j$ of the Markov chain that is
generated is such that $\theta^{(1)}_j\in
\Theta^*_j\overset{\underset{\mathrm{def}}{}}{=}
\Theta\setminus\left\{\theta_{j-1}^{(1)},\ldots,\theta_{j-1}^{(N_{j-1})}\right\}$,
then $\theta^{(i)}_j\in\Theta^*_j$ for all $i\geq1$. Indeed, the
only way for the chain to enter the set
$\left\{\theta_{j-1}^{(1)},\ldots,\theta_{j-1}^{(N_{j-1})}\right\}$
is to generate a candidate state
$\xi\in\left\{\theta_{j-1}^{(1)},\ldots,\theta_{j-1}^{(N_{j-1})}\right\}$;
however, according to Case II, such a candidate will always be
rejected. Thus, by replacing the state space
$\Theta$ by $\Theta_j^*$ and using (\ref{case1}) and
(\ref{case2}) for evaluation of
$\hat{\pi}_j^{N_{j-1}}(\theta)/\hat{\pi}_j^{N_{j-1}}(\xi)$, we are
able to calculate the acceptance probability
$\alpha_j(\xi|\theta)=\min\left\{1,\frac{\pi_j(\xi)\hat{\pi}_j^{N_{j-1}}(\theta)}{\pi_j(\theta)\hat{\pi}_j^{N_{j-1}}(\xi)}\right\}$
involved in the IMH algorithm. It is clear that the replacement of
$\Theta$ by $\Theta_j^*$ is harmless for the ergodic
properties of the Markov chain when $\Theta\subseteq\mathbb{R}^d$.
\begin{remark} One may wonder why not just use the continuous part of
$\hat{\pi}_j^{N_{j-1}}(\cdot)$ as the global proposal density within
the IMH algorithm. In other words, why not use the density
$\hat{\pi}_{j,\mathrm{cont}}^{N_{j-1}}(\cdot)$, which is
proportional to the function defined by (\ref{formaldefinition}), as
the proposal density. Indeed, in this case we would not have any
difficulties with calculating the ratio
$\hat{\pi}_j^{N_{j-1}}(\theta)/\hat{\pi}_j^{N_{j-1}}(\xi)$. The
problem is that it is not clear how to sample from
$\hat{\pi}_{j,\mathrm{cont}}^{N_{j-1}}(\cdot)$, while sampling from
$\hat{\pi}_j^{N_{j-1}}(d\theta)=\sum_{i=1}^{N_{j-1}}
   \bar{w}^{(i)}_{j-1} K_j(d\theta|\theta_{j-1}^{(i)})$ is straightforward.
\end{remark}

The above discussion leads to the following algorithm for sampling
from the distribution $\pi_j(\cdot)$:

\vspace{3mm} \hrule height 0.6pt \rule{0pt}{4mm}\centerline
{\textbf{AIMS at annealing level $j$}}\rule{0pt}{4mm}
 \hrule
 \vspace{1mm}
 \texttt{Input:}

 \hspace{0.5cm}$\vartriangleright$
 $\theta_{j-1}^{(1)},\ldots,\theta_{j-1}^{(N_{j-1})}\sim\pi_{j-1}(\cdot)$,
  samples generated at annealing level $j-1$;

 \hspace{0.5cm}$\vartriangleright$
 $\theta_j^{(1)}\in\Theta^*_j=
\Theta\setminus\left\{\theta_{j-1}^{(1)},\ldots,\theta_{j-1}^{(N_{j-1})}\right\}$,
initial state of a Markov chain;

 \hspace{0.5cm}$\vartriangleright$ $q_j(\cdot|\xi)$, symmetric
 proposal density associated with the RWMH kernel;

\hspace{0.5cm}$\vartriangleright$ $N_j$, total number of Markov
chain states to be generated.

\texttt{Algorithm:}

\hspace{0.5cm} \textbf{for} $i=1,\ldots,N_j-1$ \textbf{do}

\hspace{1.2cm} 1) Generate a global candidate state
$\xi_g\sim\hat{\pi}_j^{N_{j-1}}(\cdot)$ as follows:

\hspace{1.8cm} a. Select $k$ from $\{1,\ldots,N_{j-1}\}$ with
probabilities $\bar{w}^{(i)}_{j-1}$ given by (\ref{weights}).

\hspace{1.8cm} b. Generate a local candidate $\xi_l\sim
q_j(\cdot|\theta_{j-1}^{(k)})$.

\hspace{1.8cm} c. Accept or reject $\xi_l$ by setting

\begin{equation}\label{AcceptRejectthelocal}
    \xi_g=\left\{
            \begin{array}{ll}
              \xi_l, & \hbox{with probability } \min\left\{1,\frac{\pi_j(\xi_l)}{\pi_j(\theta_{j-1}^{(k)})}\right\}; \\
              \theta_{j-1}^{(k)}, & \hbox{with the remaining probability.}
            \end{array}
          \right.
\end{equation}

\hspace{1.2cm} 2) Update $\theta_j^{(i)}\rightarrow\theta_j^{(i+1)}$
by accepting or rejecting $\xi_g$ as follows:

\hspace{1.8cm} \textbf{if} $\xi_g=\theta_{j-1}^{(k)}$

\hspace{2.3cm} Set $\theta_j^{(i+1)}=\theta_j^{(i)}$

\hspace{1.8cm} \textbf{else}

\hspace{2.3cm} Set

\begin{equation}\label{AcceptRejecttheglobal}
    \theta_j^{(i+1)}=\left\{
                      \begin{array}{ll}
                      \xi_g, & \hbox{with probability } \min\left\{1,\frac{\pi_j(\xi_g)\hat{\pi}_j^{N_{j-1}}(\theta_j^{(i)})}{\pi_j(\theta_j^{(i)})\hat{\pi}_j^{N_{j-1}}(\xi_g)}\right\}; \\
                      \theta_j^{(i)}, & \hbox{with the remaining probability.}
                      \end{array}
                     \right.
\end{equation}

\hspace{1.8cm} \textbf{end if}

\hspace{0.5cm} \textbf{end for}

 \texttt{Output:}

 \hspace{0.5cm}$\blacktriangleright$
$\theta^{(1)}_j,\ldots,\theta^{(N_j)}_j$, $N_j$ states of a Markov
chain  with a stationary distribution $\pi_j(\cdot)$
 \vspace{1mm}
 \hrule
 \vspace{5mm}

Schematically, the AIMS algorithm at annealing level $j$ is shown in
Figure~\ref{scheme}. The proof that $\pi_j(\cdot)$ is indeed a
stationary distribution for the Markov chain generated by AIMS is
given in Section \ref{Ergodic Properties}.

\begin{remark} As usually for MCMC algorithms, the fact of convergence of a Markov chain to
its stationary distribution does not depend on the initial state;
however, the speed of convergence does. One reasonable way to chose
the initial state $\theta_j^{(1)}\in\Theta^*_j$ in practical
applications is the following: generate $\theta_j^{(1)}\sim
q_j(\cdot|\theta_{j-1}^{(k^*)})$, where $k^*=\arg\max_k
\bar{w}^{(k)}_{j-1}$, i.e. $\theta_{j-1}^{(k^*)}$ has the largest
normalized importance weight.
\end{remark}

\subsection{The full AIMS procedure}\label{Implementation}

At the zero$^{\mathrm{th}}$ annealing level, $j=0$, we generate  prior
samples $\theta_0^{(1)},\ldots,\theta_0^{(N_0)}$, which usually can
be readily drawn directly by a suitable choice of the prior distribution
$\pi_0(\cdot)$. Then, using the algorithm described in the previous
subsection, we generate samples
$\theta_1^{(1)},\ldots,\theta_1^{(N_1)}$, which are approximately
distributed according to intermediate distribution
$\pi_1(\theta)\propto\pi_0(\theta)L(\theta)^{\beta_1}$. We proceed
like this until the posterior distribution
$\pi_m(\theta)\propto\pi_0(\theta)L(\theta)^{\beta_m}$ ($\beta_m=1$)
has been sampled. To make the description of AIMS complete, we have
to explain how to choose the annealing parameters $\beta_j$, for
$j=2,\ldots,m-1$.

It is clear that the choice of the annealing parameters is very
important, since, for instance, it affects the accuracy of the
importance sampling approximation
(\ref{ImportanceSamplingApproximation}) and, therefore, the
efficiency of the whole AIMS procedure. At the same time, it is
difficult to make a rational choice of the $\beta_j$-values in
advance, since this requires some prior knowledge about the posterior
distribution, which is often not available. For this reason, we
propose an adaptive way of choosing the annealing scheme.

In importance sampling, a useful measure of degeneracy of the method
is the \textit{effective sample size} (ESS) $N^{\mathrm{eff}}$ introduced in
\cite{Kong} and \cite{Liu2}. The ESS measures how similar the
importance sampling distribution $\pi_{j-1}(\cdot)$ is to the
target distribution $\pi_{j}(\cdot)$. Suppose $N_{j-1}$
independent samples
$\theta_{j-1}^{(1)},\ldots,\theta_{j-1}^{(N_{j-1})}$ are generated
from $\pi_{j-1}(\cdot)$, then the ESS of these samples is defined as
\begin{equation}\label{ESS}
    N^{\mathrm{eff}}_{j-1}=\frac{N_{j-1}}{1+\mathrm{var}_{\pi_{j-1}}[w]}=\frac{N_{j-1}}{\mathbb{E}_{\pi_{j-1}}[w^2]},
\end{equation}
where $w(\theta)=\pi_j(\theta)/\pi_{j-1}(\theta)$. The ESS can be
interpreted as implying that $N_{j-1}$ weighted samples $(\theta_{j-1}^{(1)},
w_{j-1}^{(1)}),\ldots,(\theta_{j-1}^{(N_{j-1})},
w^{(N_{j-1})}_{j-1})$ are worth $N^{\mathrm{eff}}_{j-1}(\leq
N_{j-1}$) i.i.d. samples drawn from the target distribution
$\pi_j(\cdot)$. One cannot evaluate the ESS exactly but an estimate
$\hat{N}^{\mathrm{eff}}_{j-1}$ of $N^{\mathrm{eff}}_{j-1}$ is given
by
\begin{equation}\label{ESSestimate}
\hat{N}^{\mathrm{eff}}_{j-1}(\bar{w}_{j-1})=\frac{1}{\sum_{i=1}^{N_{j-1}}(\bar{w}_{j-1}^{(i)})^2},
\end{equation}
where
$\bar{w}_{j-1}=(\bar{w}_{j-1}^{(1)},\ldots,\bar{w}_{j-1}^{(N_{j-1})})$
and $\bar{w}_{j-1}^{(i)}$ is the normalized importance weight of
$\theta_{j-1}^{(i)}$.

At annealing level $j$, when $\beta_{j-1}$ is already known,
the problem is to define $\beta_j$. Let
$\gamma=\hat{N}^{\mathrm{eff}}_{j-1}/N_{j-1}\in(0,1)$ be a
prescribed threshold that characterizes the ``quality'' of the
weighted sample (the larger $\gamma$ is, the ``better'' the weighted
sample is). Then we obtain the following equation:
\begin{equation}\label{equation on betaj}
\sum_{i=1}^{N_{j-1}}(\bar{w}_{j-1}^{(i)})^2=\frac{1}{\gamma N_{j-1}}
\end{equation}
Observe that this equation can be expressed as an equation for $\beta_j$ by using (\ref{weights}):
\begin{equation}\label{equation on betaj 2}
\frac{\sum_{i=1}^{N_{j-1}}L(\theta_{j-1}^{(i)})^{2(\beta_j-\beta_{j-1})}}{\left(\sum_{i=1}^{N_{j-1}}L(\theta_{j-1}^{(i)})^{\beta_j-\beta_{j-1}}\right)^2}
=\frac{1}{\gamma N_{j-1}}
\end{equation}
Solving this equation for $\beta_j$ gives us the value of the annealing
parameter at level $j$.

%The advantage of the described adaptive approach is that it does not
%involve

\begin{remark} Note that when $j\geq2$, the $\theta_{j-1}^{(1)},\ldots,\theta_{j-1}^{(N_{j-1})}$ are
generated by the Markov chain sampler described in the previous
subsection and therefore are not independent. This means that,
because of the autocorrelations produced by the Markov chain used,
the ``true'' ESS of this sample is, in fact, smaller than the one
given by (\ref{ESS}). This is useful to remember when choosing
$\gamma$. Also, this is another reason to select the prior distribution $\pi_0(\cdot)$ so that samples
can be generated independently at the start of each
AIMS run.
\end{remark}

Combining the AIMS algorithm at a given annealing level with the
described adaptive annealing scheme gives rise to the
following procedure.

\vspace{3mm} \hrule height 0.6pt \rule{0pt}{4mm}\centerline
{\textbf{The AIMS procedure}}\rule{0pt}{4mm}
 \hrule
 \vspace{1mm}
 \texttt{Input:}

\hspace{0.5cm}$\vartriangleright$ $\gamma$, threshold for the
effective sample size (ESS);

\hspace{0.5cm}$\vartriangleright$ $N_0, N_1, \ldots$, where $N_j$ is
the total number of Markov chain states to be generated

\hspace{0.85cm} at annealing level $j$;

\hspace{0.5cm}$\vartriangleright$ $q_1(\cdot|\xi), q_2(\cdot|\xi),
\ldots$, where $q_j(\cdot|\xi)$ is the symmetric  proposal density
associated with

\hspace{0.85cm} the RWMH kernel at annealing level $j$.

\texttt{Algorithm:}

\hspace{0.5cm} Set $j=0$, current annealing level.

\hspace{0.5cm} Set $\beta_0=0$, current annealing parameter.

\hspace{0.5cm} Sample $\theta_0^{(1)},\ldots,
\theta_0^{(N_0)}\stackrel{i.i.d}{\sim}\pi_0(\cdot)$.

\hspace{0.5cm} Calculate
$\bar{W}_0^{(i)}=\frac{L(\theta_{0}^{(i)})^{1-\beta_{0}}}{\sum_{i=1}^{N_{0}}L(\theta_{0}^{(i)})^{1-\beta_{0}}}$,
$i=1,\ldots,N_0$.

\hspace{0.5cm} Calculate the ESS
$\hat{N}^{\mathrm{eff}}_0=\hat{N}^{\mathrm{eff}}_0(\bar{W}_0)$
 using (\ref{ESSestimate}), which measures how similar the

\hspace{0.5cm} prior distribution $\pi_{0}(\cdot)$ is to the target
posterior distribution $\pi(\cdot)$.

\hspace{0.5cm} \textbf{while}
$\hat{N}^{\mathrm{eff}}_j/N_{j}<\gamma$ \textbf{do}

\hspace{1.7cm} Find $\beta_{j+1}$ from equation (\ref{equation on
betaj 2}).

\hspace{1.7cm} Calculate normalized importance weights
$\bar{w}^{(i)}_j$, $i=1,\ldots,N_j$ using (\ref{weights}).

\hspace{1.7cm} Generate  a Markov chain $\theta_{j+1}^{(1)},\ldots,
\theta_{j+1}^{(N_{j+1})}$ with the stationary distribution

\hspace{1.7cm}  $\pi_{j+1}(\cdot)$ using the AIMS algorithm at
annealing level $j+1$.

%\begin{equation}\label{AIMS level j+1}
%\mathrm{AIMS}_{j+1}[]:\theta_{j}^{(1)},\ldots,
%\theta_{j}^{(N_{j})}\longrightarrow\theta_{j+1}^{(1)},\ldots,
%\theta_{j+1}^{(N_{j+1})}
%\end{equation}

\hspace{1.7cm} Calculate
$\bar{W}_{j+1}^{(i)}=\frac{L(\theta_{j+1}^{(i)})^{1-\beta_{j+1}}}{\sum_{i=1}^{N_{j+1}}L(\theta_{j+1}^{(i)})^{1-\beta_{j+1}}}$,
$i=1,\ldots,N_{j+1}$.

\hspace{1.7cm} Calculate the ESS
$\hat{N}^{\mathrm{eff}}_{j+1}=\hat{N}^{\mathrm{eff}}_{j+1}(\bar{W}_{j+1})$
 using (\ref{ESSestimate}), which measures how

\hspace{1.7cm} similar the intermediate distribution
$\pi_{j+1}(\cdot)$ is to the posterior $\pi(\cdot)$.

\hspace{1.7cm} Increment $j$ to $j+1$.

\hspace{0.5cm} \textbf{end while}

\hspace{0.5cm} Set $\beta_{j+1}=1$, current annealing parameter.

\hspace{0.5cm} Set $m=j+1$, the total number of distributions in the
annealing scheme.

\hspace{0.5cm} Set $\bar{w}^{(i)}_{m-1}=\bar{W}^{(i)}_{m-1}$,
$i=1,\ldots,N_{m-1}$.

\hspace{0.61cm}Generate  a Markov chain $\theta_{m}^{(1)},\ldots,
\theta_{m}^{(N_{m})}$ with the stationary distribution

\hspace{0.5cm}  $\pi_{m}(\cdot)=\pi(\cdot)$ using the AIMS algorithm
at annealing level $m$.

  \texttt{Output:}

 \hspace{0.5cm}$\blacktriangleright$
$\theta^{(1)}_m,\ldots,\theta^{(N_m)}_m\dot{\sim}\pi(\cdot)$, samples that are approximately distributed according

\hspace{0.9cm} to the posterior distribution.
 \vspace{1mm}
 \hrule
 \vspace{5mm}

\subsection{Implementation issues}\label{impl issues}

As it follows from the description, the AIMS procedure has the
following parameters: $\gamma$, the threshold for the effective
sample size; $N_j$, the length of a Markov chain generated at
annealing level $j=1,\ldots,m$; and $q_j(\cdot|\xi)$, the symmetric
proposal density associated with the RWMH kernel at level
$j=1,\ldots,m$. Here, we discuss the choice of these parameters and
how this choice affects the efficiency of AIMS.

First of all, it is absolutely clear that, as for any Monte Carlo
method, the larger the number of generated samples is, the more
accurate the corresponding estimates of (\ref{E[h]}) are. However, we would like to
highlight the difference between the roles of $N_{j-1}$ and $N_j$ at
annealing level $j$. While $N_j$ is directly related to the
convergence of the chain $\theta_j^{(1)},\ldots,\theta_j^{(N_j)}$ to
its stationary distribution $\pi_j(\cdot)$, $N_{j-1}$ affects this
convergence implicitly through the global proposal distribution
$\hat{\pi}_j^{N_{j-1}}(\cdot)$: the larger $N_{j-1}$, the more
accurate approximation (\ref{ImportanceSamplingApproximation}) is,
and, therefore, the less correlated $\theta_{j}^{(1)},\ldots,
\theta_{j}^{(N_j)}$ are. When $N_{j-1}\rightarrow\infty$, samples
$\theta_{j}^{(1)},\ldots, \theta_{j}^{(N_j)}$ become independent
draws from $\pi_j(\cdot)$, hence the name of the algorithm. Thus, if
we increase $N=N_{j-1}=N_j$, the effect is twofold: first, the
sample size increases thereby increasing the effective number of independent samples at the $j^{\rm{th}}$ level (typical for any Monte Carlo method);
second, the samples become less correlated (a
useful feature of AIMS), again increasing the effective number of independent samples. As a result of these two effects,
increasing $N$ has a strong influence on the effective number of independent posterior samples and so strongly reduces
the variance of the estimator for (\ref{E[h]}).

Suppose now that we are at the last annealing level and generating a Markov chain
$\theta_{m}^{(1)},\ldots,
\theta_{m}^{(N_{m})}$ with the stationary distribution $\pi_{m}(\cdot)=\pi(\cdot)$. We will refer to this chain as the posterior Markov chain.
A critical question faced by users of MCMC methods is how to
determine when it is safe to stop sampling from the posterior distribution and use samples $\theta_{m}^{(1)},\ldots,
\theta_{m}^{(N_{m})}$ for estimation. In other words, how large
should $N_m$ be? One possible solution of this ``convergence assessment problem'' is to use one of the numerous published diagnostic techniques;
for example, see \cite{CowlesCarlin} for a comparative review of MCMC convergence diagnostics.
Unfortunately, none of the published diagnostics allows one to say with certainty that a finite sample from an MCMC algorithm is representative
of an underlying stationary distribution.
A more empirical approach for assessing convergence is to run several posterior Markov chains
 $\theta_{k,m}^{(1)},\ldots, \theta_{k,m}^{(N_{m})}$, $k=1,\ldots,K$, in parallel and
 monitor the corresponding estimators $\hat{h}_1,\ldots,\hat{h}_K$ of $\mathbb{E}_{\pi}[h]$.
A stopping rule for convergence is then
\begin{equation}\label{stopping rule}
    \max_{1\leq i<j\leq K}|\hat{h}_i-\hat{h}_j|<\varepsilon,
\end{equation}
where $\varepsilon$ is  a minimum precision requirement. It is important to emphasise, though, that rule (\ref{stopping rule}),
although easy-to-understand and easy-to-implement,
does not assure convergence of the chains (especially if $\pi(\cdot)$ is multi-modal):
``the potential for problems with multiple modes exists whenever there is no theoretical guarantee that the distribution is unimodal'' \cite{Neal_AIS}.

The threshold $\gamma$ affects the speed of annealing. If $\gamma$ is very small, i.e. close to zero,
then AIMS will have very few intermediate distributions interpolating between the prior and
posterior distributions, and this will lead to inaccurate results for a moderate number of samples. On the other hand, if $\gamma$ is very large,
i.e. close to one, then AIMS will have too many intermediate distributions, which will make the algorithm
computationally very expensive.

The proposed method for finding $\beta_j$-values is based on the ESS,
and $\beta_j$ is defined from equation (\ref{equation on betaj}) (or, equivalently, from (\ref{equation on betaj 2})).
A similar adaptive approach for defining an annealing scheme was
proposed in \cite{Ching}. It is based on the coefficient of variation (COV)
of the importance weights (\ref{weights}). More precisely, the equation for $\beta_j$ is given by
\begin{equation}\label{COV-creterion}
\frac{\sqrt{\frac{1}{N_{j-1}}\sum_{i=1}^{N_{j-1}}\left(w_{j-1}^{(i)}-\frac{1}{N_{j-1}}\sum_{i=1}^{N_{j-1}}w_{j-1}^{(i)}\right)^2}}{\frac{1}{N_{j-1}}\sum_{i=1}^{N_{j-1}}w_{j-1}^{(i)}}=\delta,
\end{equation}
where $\delta>0$ is a prescribed threshold. It is easy to show that the ESS-criterion (\ref{equation on betaj})
and the COV-criterion (\ref{COV-creterion}) are mathematically equivalent; in fact, $\hat{N}^{\mathrm{eff}}_{j-1}=N_{j-1}/(1+\delta^2)$.
We prefer to use the former criterion since $\gamma$ has a clear meaning:
it is the factor by which the (essential) sample size of the weighted sample is reduced as a penalty for sampling from the importance sampling density
instead of the target distribution. It has been found in \cite{Ching} that $\delta=1$ is usually a
reasonable choice of the threshold. This corresponds to $\gamma=1/2$.
Our simulation results (see Section \ref{examples-section}) also show that annealing schemes with
$\gamma$ around $1/2$ yield good efficiency.

The choice of the local proposal density $q_j(\cdot|\xi)$ associated with the RWMH kernel determines
the ergodic properties of the Markov chain generated by AIMS at level
$j$; it also determines how efficiently the chain explores local neighborhoods of samples $\theta_{j-1}^{(1)},\ldots,\theta_{j-1}^{(N_{j-1})}$
generated at the previous level. This makes the choice of $q_j(\cdot|\xi)$  very important.

It has been observed by many researchers that the efficiency of Metropolis-Hastings based MCMC methods
is not sensitive to the type of the proposal density; however, it strongly depends on its variance (e.g. \cite{Gelman,AuBeck}).
For this reason, we suggest using a Gaussian density as the local proposal:
\begin{equation}\label{Gaussian proposal}
    q_j(\theta|\xi)=\mathcal{N}(\theta|\xi,c_j^2\mathbb{I}),
\end{equation}
where $\xi$ and $c_j^2\mathbb{I}$ are the mean and diagonal covariance matrix, respectively.
The scaling parameter $c_j^2$ determines the ``spread'' of the local proposal distribution.
In Section \ref{Ergodic Properties}, we prove (Theorem \ref{theorem3}) that, under certain conditions, the acceptance rate $\mathcal{\bar{A}}_j$
(i.e. the expected probability of having a proper Markov transition $\theta_j^{(i)}$ to $\theta_j^{(i+1)}\neq\theta_j^{(i)}$)
satisfies $\mathcal{\bar{A}}_j\geq \frac{1}{M}$, where constant $M$ depends on $q_j(\cdot|\xi)$ and, therefore, on $c_j^2$.
This result can be potentially used for finding an optimal $c_j^2$ that would minimize $M$.
Alternatively, a more empirical way of choosing the scaling factor consists of adjusting $c_j^2$ based on the estimated acceptance rate.
This works as follows: first, choose an initial value for the scaling factor, $c_{j,0}^2$, and estimate the corresponding
acceptance rate $\mathcal{\bar{A}}_j(c_{j,0}^2)$ based on $N_j$ generated Markov states, then modify $c_{j,0}^2$ to obtain
an increase in $\mathcal{\bar{A}}_j$. Whether this optimization in $c_j^2$ is useful depends on whether the accuracy of the estimator that is achieved
compensates for the additional computational cost. Finally, note that our simulation results show (see Section \ref{examples-section}) that,
as $j$ increases, the corresponding optimal scaling factor $c_j^2$ decreases slightly. This observation coincides with intuition, since when
$j$ increases, the intermediate distributions $\pi_j(\cdot)$ become more concentrated.

In the following section we establish the ergodic properties of the
Markov chains generated by AIMS.

\section{Ergodic properties of AIMS}\label{Ergodic Properties}
Since the discussion in Subsection \ref{level j}, which motivated
AIMS at annealing level $j$, involved delta functions and formal
equalities (\ref{formaldefinition}) and (\ref{formaldefinition2}),
we cannot simply rely on the convergence of the IMH algorithm in verification of AIMS;
a rigorous proof is needed. First we prove that the described
algorithm indeed generates a Markov chain with a stationary
distribution $\pi_j(\cdot)$. We also explain that when the proposal
density $q_j(\cdot|\xi)$ is reasonably chosen, $\pi_j(\cdot)$ is the
unique (and, therefore, limiting) stationary distribution of the
corresponding Markov chain.

\begin{theorem}\label{theorem1} Let $\theta^{(1)}_j,\theta^{(2)}_j,\ldots$ be the Markov chain on $\Theta_j^*=\Theta\setminus\left\{\theta_{j-1}^{(1)},\ldots,\theta_{j-1}^{(N_{j-1})}\right\}$ generated
by the AIMS algorithm at annealing level $j$, then $\pi_j(\cdot)$ is
a stationary
distribution of the Markov chain.%, i.e. if $\theta^{(i)}_j$ is
%distributed according to $\pi_j(\cdot)$, then so is
%$\theta^{(i+1)}_j$.
\end{theorem}

\begin{proof} Let $\mathcal{K}_j(\cdot|\cdot)$ denote the transition kernel of the Markov chain generated by
AIMS at annealing level $j$. From the discription of the algorithm
it follows that $\mathcal{K}_j(\cdot|\cdot)$ has the following form:

\begin{equation}\label{AIMStransitionkernel}
\begin{split}
\mathcal{K}_j(d\xi|\theta)&=\sum_{i=1}^{N_{j-1}} \bar{w}^{(i)}_{j-1}
q_j(\xi|\theta_{j-1}^{(i)})\min\left\{1,\frac{\pi_j(\xi)}{\pi_j(\theta_{j-1}^{(i)})}\right\}
\min\left\{1,\frac{\pi_j(\xi)\hat{\pi}_{j}^{N_{j-1}}(\theta)}{\pi_j(\theta)\hat{\pi}_{j}^{N_{j-1}}(\xi)}\right\}d\xi\\
&+(1-\mathcal{A}_j(\theta))\delta_{\theta}(d\xi),
\end{split}
\end{equation}
where $\mathcal{A}_j(\theta)$ is the probability of having a proper
transition $\theta$ to $\Theta_j^*\setminus\{\theta\}$:
\begin{equation}\label{AIMSacceptanceprob}
    \mathcal{A}_j(\theta)=\int_{\Theta_j^*} \sum_{i=1}^{N_{j-1}}
    \bar{w}^{(i)}_{j-1}
q_j(\xi|\theta_{j-1}^{(i)})\min\left\{1,\frac{\pi_j(\xi)}{\pi_j(\theta_{j-1}^{(i)})}\right\}
\min\left\{1,\frac{\pi_j(\xi)\hat{\pi}_{j}^{N_{j-1}}(\theta)}{\pi_j(\theta)\hat{\pi}_{j}^{N_{j-1}}(\xi)}\right\}d\xi
\end{equation}

A sufficient condition for $\pi_j(\cdot)$ to be a stationary
distribution is for $\mathcal{K}_j(\cdot|\cdot)$ to satisfy the
detailed balance condition:
\begin{equation}\label{DBC}
    \pi_j(d\theta)\mathcal{K}_j(d\xi|\theta)=\pi_j(d\xi)\mathcal{K}_j(d\theta|\xi)
\end{equation}
Without loss of generality, we assume that $\theta\neq\xi$, since
otherwise (\ref{DBC}) is trivial. In this case
$\mathcal{K}_j(d\xi|\theta)$ is given by the first term in
(\ref{AIMStransitionkernel}), since the second term vanishes. Thus,
all we need to prove is that function
\begin{equation}\label{expression}
\mathcal{E}(\theta,\xi)\overset{\underset{\mathrm{def}}{}}{=}\pi_j(\theta)\sum_{i=1}^{N_{j-1}}
\bar{w}^{(i)}_{j-1}
q_j(\xi|\theta_{j-1}^{(i)})\min\left\{1,\frac{\pi_j(\xi)}{\pi_j(\theta_{j-1}^{(i)})}\right\}
\min\left\{1,\frac{\pi_j(\xi)\hat{\pi}_{j}^{N_{j-1}}(\theta)}{\pi_j(\theta)\hat{\pi}_{j}^{N_{j-1}}(\xi)}\right\}
\end{equation}
is symmetric with respect to permutation $\theta\leftrightarrow\xi$,
for all $\theta,\xi\in\Theta_j^*$. Taking into account
(\ref{formaldefinition}) and a simple fact that
$a\min\{1,b/a\}=b\min\{1,a/b\}$ for all $a,b>0$, we have:
\begin{equation}\label{expression2}
\begin{split}
\mathcal{E}(\theta,\xi)&=
\pi_j(\theta)\hat{\pi}_{j}^{N_{j-1}}(\xi)\min\left\{1,\frac{\pi_j(\xi)\hat{\pi}_{j}^{N_{j-1}}(\theta)}{\pi_j(\theta)\hat{\pi}_{j}^{N_{j-1}}(\xi)}\right\}\\
&=\pi_j(\xi)\hat{\pi}_{j}^{N_{j-1}}(\theta)\min\left\{1,\frac{\pi_j(\theta)\hat{\pi}_{j}^{N_{j-1}}(\xi)}{\pi_j(\xi)\hat{\pi}_{j}^{N_{j-1}}(\theta)}\right\}=\mathcal{E}(\xi,\theta)
\end{split}
\end{equation}
This proves that $\pi_j(\cdot)$ is a stationary distribution of the
AIMS Markov chain.
\end{proof}

A stationary distribution is unique and is the limiting distribution
for a Markov chain, if the chain is aperiodic and irreducible (see,
for example, \cite{Tierney}). In the case of AIMS, aperiodicity is
guaranteed by the fact that the probability of having a repeated
sample $\theta^{(i+1)}_j=\theta^{(i)}_j$ is not zero: for example,
if the local candidate state $\xi_l$ is rejected in step 1c, then we
automatically have $\theta^{(i+1)}_j=\theta^{(i)}_j$. A Markov chain
with stationary distribution $\pi(\cdot)$ is irreducible if, for any
initial state, it has positive probability of entering any set to
which $\pi(\cdot)$ assigns positive probability. It is clear that if
the proposal distribution $q_j(\cdot|\xi)$ is ``standard'' (e.g.
Gaussian, uniform, log-normal, etc), then AIMS generates an irreducible
Markov chain. In this case, $\pi_j(\cdot)$ is therefore the unique
stationary distribution of the AIMS Markov chain, and for every
$\theta\in\Theta_j^*$
\begin{equation}\label{limitingproperty}
\lim_{n\rightarrow\infty}\|\mathcal{K}^n_j(\cdot|\theta)-\pi_j(\cdot)\|_{\mathrm{TV}}=0,
\end{equation}
with $\|\cdot\|_{\mathrm{TV}}$ denoting the total variation distance.
Recall that the total variation distance between two measures
$\mu_1(\cdot)$ and $\mu_2(\cdot)$ on $\Theta$  is defined as
$\|\mu_1(\cdot)-\mu_1(\cdot)\|_{\mathrm{TV}}=\sup_{A\subset\Theta}|\mu_1(A)-\mu_2(A)|$.
In a simulation setup, the most important consequence of convergence
property (\ref{limitingproperty}) is, of course, that the sample
mean converges to the expectation of a measurable function of interest almost
surely:
\begin{equation}\label{convergencetothemean}
    \lim_{N_j\rightarrow\infty}\frac{1}{N_j}\sum_{i=1}^{N_j}h(\theta_j^{(i)})=\int_{\Theta}h(\theta)\pi_j(\theta)d\theta
\end{equation}

Convergence (\ref{limitingproperty}) ensures the proper behavior of
the AIMS chain $\theta_j^{(1)},\theta_j^{(2)},\ldots$ regardless of
the initial state $\theta_j^{(1)}$. A more detailed description of
convergence properties involves the study of the speed of
convergence of $\mathcal{K}_j^n(\cdot|\theta)$ to $\pi_j(\cdot)$. %, where $\mathcal{K}_j^n(\cdot|\theta)$ is the $n$-step ahead transition kernel
%for the $j^{\mathrm{th}}$ annealing level.
Evaluation (or estimation) of this speed is very important for any
MCMC algorithm, since it relates to a stopping rule for this
algorithm: the higher the speed of convergence
$\mathcal{K}_j^n(\cdot|\theta)\rightarrow\pi_j(\cdot)$, the less
samples are need to obtain an accurate estimate in
(\ref{convergencetothemean}). Recall, following \cite{MeynTweedie},
that a chain $\theta^{(1)},\theta^{(2)},\ldots$ is called
\textit{uniformly ergodic} if
\begin{equation}\label{uniform ergodicity}
    \lim_{n\rightarrow\infty}\sup_{\theta\in\Theta}\|\mathcal{K}^n(\cdot|\theta)-\pi(\cdot)\|_{\mathrm{TV}}=0
\end{equation}
The property of uniform ergodicity is stronger than
(\ref{limitingproperty}), since it guarantees that the speed of
convergence is uniform over the whole space. Moreover, a Markov
chain is uniformly ergodic if and only if there exist $r>1$ and
$R<\infty$ such that for all $\theta\in\Theta$
\begin{equation}\label{uniformgeometricrate}
    \|\mathcal{K}^n(\cdot|\theta)-\pi(\cdot)\|_{\mathrm{TV}}\leq
    Rr^{-n},
\end{equation}
that is, the convergence in (\ref{uniform ergodicity}) takes place
at uniform geometric rate \cite{MeynTweedie}.

\begin{theorem} If there exists a constant $M$ such that for all $\theta\in\Theta^*_j$
\begin{equation}\label{condition}
\pi_j(\theta)\leq M \hat{\pi}_j^{N_{j-1}}(\theta),
\end{equation}
then the AIMS algorithm at annealing level $j$ produces a uniformly
ergodic chain and
\begin{equation}\label{speed of convergence}
   \|\mathcal{K}_j^n(\cdot|\theta)-\pi_j(\cdot)\|_{\mathrm{TV}}\leq
    \left(1-\frac{1}{M}\right)^{n}
\end{equation}
\end{theorem}
\begin{proof}
To prove the first part of the theorem we will need the notion of a
\textit{small set} \cite{MeynTweedie}. A set $A\subset\Theta$ is
called a small set if there exists an integer $m>0$ and a
non-trivial measure $\mu_m$ on $\Theta$, such that for all
$\theta\in A$, $B\subset\Theta$:
\begin{equation}\label{small set}
    \mathcal{K}^m(B|\theta)\geq \mu_m(B)
\end{equation}
In this case we say that $A$ is $\mu_m$-small. It can be shown
\cite{MeynTweedie} that a Markov chain is uniformly ergodic if and
only if its state space is $\mu_m$-small for some $m$. Thus, to
prove the theorem, it is enough to show that $\Theta^*_j$ is a small
set.

If (\ref{condition}) is satisfied, than the following holds for
transition kernel (\ref{AIMStransitionkernel}) for $\theta\in\Theta^*_j$ and $B\subset\Theta^*_j$:
\begin{equation}\label{bound for kernel}
\begin{split}
\mathcal{K}_j(B|\theta)&\geq\int_B\sum_{i=1}^{N_{j-1}}
\bar{w}^{(i)}_{j-1}
q_j(\xi|\theta_{j-1}^{(i)})\min\left\{1,\frac{\pi_j(\xi)}{\pi_j(\theta_{j-1}^{(i)})}\right\}
\min\left\{1,\frac{\pi_j(\xi)\hat{\pi}_{j}^{N_{j-1}}(\theta)}{\pi_j(\theta)\hat{\pi}_{j}^{N_{j-1}}(\xi)}\right\}d\xi\\
&=\int_B
\hat{\pi}_j^{N_{j-1}}(\xi)\min\left\{1,\frac{\pi_j(\xi)\hat{\pi}_{j}^{N_{j-1}}(\theta)}{\pi_j(\theta)\hat{\pi}_{j}^{N_{j-1}}(\xi)}\right\}d\xi\\
&=\int_B\min\left\{\hat{\pi}_j^{N_{j-1}}(\xi),\pi_j(\xi)\frac{\hat{\pi}_{j}^{N_{j-1}}(\theta)}{\pi_j(\theta)}\right\}d\xi\\
&\geq\int_B\min\left\{\hat{\pi}_j^{N_{j-1}}(\xi),\frac{\pi_j(\xi)}{M}\right\}d\xi=\frac{1}{M}\int_B\pi_j(\xi)d\xi=\frac{1}{M}\pi_j(B)
\end{split}
\end{equation}
%for any $\theta\in\Theta_j^*$.
The sample space $\Theta^*_j$ is
therefore $\frac{\pi_j}{M}$-small, and the corresponding Markov chain is
uniformly ergodic.

To prove bound (\ref{speed of convergence}), first observe, using
(\ref{bound for kernel}), that
\begin{equation}\label{n=1}
\|\mathcal{K}_j(\cdot|\theta)-\pi_j(\cdot)\|_{\mathrm{TV}}=\sup_A|\mathcal{K}_j(A|\theta)-\pi_j(A)|\leq\sup_A|\pi_j(A)-\frac{1}{M}\pi_j(A)|=1-\frac{1}{M}
\end{equation}
For $n>1$, using the Chapman-Kolmogorov equation
$\mathcal{K}^{m+n}(A|\theta)=\int_\Theta
\mathcal{K}^m(A|\xi)\mathcal{K}^n(d\xi|\theta)$ and stationarity of
$\pi_j(\cdot)$ with respect to $\mathcal{K}_j(\cdot|\cdot)$, we
have:
\begin{equation}\label{n>1}
\begin{split}
    \|\mathcal{K}_j^n(\cdot|\theta)-\pi_j(\cdot)\|_{\mathrm{TV}}&=\sup_A|\mathcal{K}_j^n(A|\theta)-\pi_j(A)|\\
    &=\sup_A\left|\int_{\Theta_j^*}\mathcal{K}_j(A|\xi)\mathcal{K}_j^{n-1}(d\xi|\theta)-\int_{\Theta_j^*}\mathcal{K}_j(A|\xi)\pi_j(\xi)d\xi\right|\\
    &=\sup_A\left|\int_{\Theta_j^*}\mathcal{K}_j(A|\xi)\left[\mathcal{K}_j^{n-1}(d\xi|\theta)-\pi_j(\xi)d\xi\right]\right|\\
    &=\sup_A\left|\int_{\Theta_j^*}\left[\mathcal{K}_j(A|\xi)-\pi_j(A)\right]\left[\mathcal{K}_j^{n-1}(d\xi|\theta)-\pi_j(\xi)d\xi\right]\right|,
 \end{split}
\end{equation}
where the last equality follows from the fact that
$\int_{\Theta_j^*}\mathcal{K}_j^{n-1}(d\xi|\theta)=\int_{\Theta_j^*}\pi_j(\xi)d\xi=1$.
Finally, we obtain:
\begin{equation}\label{final}
\begin{split}
\|\mathcal{K}_j^n(\cdot|\theta)-\pi_j(\cdot)\|_{\mathrm{TV}}&
\leq\sup_B\sup_A\left|\int_{B}\left[\mathcal{K}_j(A|\xi)-\pi_j(A)\right]\left[\mathcal{K}_j^{n-1}(d\xi|\theta)-\pi_j(\xi)d\xi\right]\right|\\
&\leq\sup_B\left|\int_B\sup_A|\mathcal{K}_j(A|\xi)-\pi_j(A)|\left[\mathcal{K}_j^{n-1}(d\xi|\theta)-\pi_j(\xi)d\xi\right]\right|\\
&=
\|\mathcal{K}_j(\cdot|\theta)-\pi_j(\cdot)\|_{\mathrm{TV}}\cdot\|\mathcal{K}_j^{n-1}(\cdot|\theta)-\pi_j(\cdot)\|_{\mathrm{TV}}%\\&
\leq
\left(1-\frac{1}{M}\right)^{n}
\end{split}
\end{equation}
\end{proof}

\begin{remark} Note that if there exists a constant $M$ such that (\ref{condition}) holds for all $\theta\in\Theta_j^*$, then $M>1$ automatically.
\end{remark}

\begin{corollary} If $\Theta\subset\mathbb{R}^d$ is a compact set and $q_j(\cdot|\xi)$ is a Gaussian distribution centered at $\xi$,
then the AIMS algorithm at annealing level $j$ produces a uniformly
ergodic chain and (\ref{speed of convergence}) holds with $M$ given
by \begin{equation}\label{M}
    M=\left(\sum_{i=1}^{N_{j-1}}
\bar{w}^{(i)}_{j-1}\frac{\min_{\theta\in\Theta}q_j(\theta|\theta_{j-1}^{(i)})}{\max_{\theta\in\Theta}\pi_j(\theta)}\right)^{-1}
\end{equation}\label{corollary}
\end{corollary}
\begin{proof}
Let us show that in this case condition (\ref{condition}) is always
fulfilled. For any $\theta\in\Theta_j^*$ we have:
\begin{equation}\label{pi hat}
\begin{split}
\hat{\pi}_j^{N_{j-1}}(\theta)&=\sum_{i=1}^{N_{j-1}}
\bar{w}^{(i)}_{j-1}
q_j(\theta|\theta_{j-1}^{(i)})\min\left\{1,\frac{\pi_j(\theta)}{\pi_j(\theta_{j-1}^{(i)})}\right\}\\
&=\sum_{i=1}^{N_{j-1}} \bar{w}^{(i)}_{j-1}
q_j(\theta|\theta_{j-1}^{(i)})\frac{\pi_j(\theta)}{\pi_j(\theta_{j-1}^{(i)})}\min\left\{1,\frac{\pi_j(\theta_{j-1}^{(i)})}{\pi_j(\theta)}\right\}\\
&\geq \pi_j(\theta) \sum_{i=1}^{N_{j-1}}
\bar{w}^{(i)}_{j-1}\frac{\min_{\theta\in\Theta}q_j(\theta|\theta_{j-1}^{(i)})}{\pi_j(\theta_{j-1}^{(i)})}
\min\left\{1,\frac{\pi_j(\theta_{j-1}^{(i)})}{\max_{\theta\in\Theta}\pi_j(\theta)}\right\}\\
&=\pi_j(\theta) \sum_{i=1}^{N_{j-1}}
\bar{w}^{(i)}_{j-1}\frac{\min_{\theta\in\Theta}q_j(\theta|\theta_{j-1}^{(i)})}{\max_{\theta\in\Theta}\pi_j(\theta)}
\end{split}
\end{equation}
Thus, (\ref{condition}) holds with $M$ given by (\ref{M}).
\end{proof}

\begin{remark} Note than the assumption of compactness of the sample space $\Theta$ is not very restrictive and
is typically satisfied in most Bayesian statistics problems. Indeed,
to fulfill this condition, it is enough to take a prior distribution
$\pi_0(\cdot)$ with compact support. Next, it is clear from the proof, that the conclusion of Corollary
\ref{corollary} holds for different ``reasonable'' (not only
Gaussian) proposal distributions $q_j(\cdot|\xi)$. Therefore, the
AIMS algorithm will produce a uniformly ergodic Markov chain in many
practical cases.
\end{remark}

It has been recognized for a long time that, when using an MCMC
algorithm, it is useful to monitor its acceptance rate
$\mathcal{\bar{A}}$, i.e. expected probability of having a proper
Markov jump
$\theta^{(i)}$ to $\theta^{(i+1)}\neq\theta^{(i)}$. While
in the case of the RWMH algorithm, the finding of the optimal acceptance
rate is a difficult problem: neither high nor low
$\mathcal{\bar{A}}$ is good \cite{Gelman}; for IMH the picture is rather simple:
the higher $\mathcal{\bar{A}}$, the better \cite{RobCas}. Since
AIMS is based on the IMH algorithm, their properties are very
similar. In particular, one should aim for the highest possible
acceptance rate of the global candidate state $\xi_g$ when
implementing AIMS.

We finish this section with a result that provides bounds for the
acceptance rate of the AIMS algorithms. These bounds can be useful for finding the optimal implementation parameters.

\begin{theorem}\label{theorem3} Let $\mathcal{\bar{A}}_j$ be the expected
probability of having a proper Markov transition associated with the
AIMS algorithm at annealing level $j$. Then
\begin{equation}\label{<}
    \mathcal{\bar{A}}_j\leq
    \sum_{i=1}^{N_{j-1}}\bar{w}_{j-1}^{(i)}a_j(\theta_{j-1}^{(i)}),
\end{equation}
where $a_j(\theta_{j-1}^{(i)})$ is probability (\ref{acceptProb})
associated with having a proper transition under the RWMH transition kernel (\ref{RWMHkernel}). If
(\ref{condition}) holds, then
\begin{equation}\label{>}
    \mathcal{\bar{A}}_j\geq \frac{1}{M}
\end{equation}
\end{theorem}
\begin{proof} For every $\theta\in\Theta_j^*$, the probability $\mathcal{A}_j(\theta)$ of transition
$\theta$ to $\Theta_j^*\setminus\{\theta\}$ is given by
(\ref{AIMSacceptanceprob}). For its expected value we have:
\begin{equation}\label{1stbound}
\begin{split}
\mathcal{\bar{A}}_j&=\int_{\Theta_j^*}\pi_j(\theta)\mathcal{A}_j(\theta)d\theta\\
&=\int_{\Theta_j^*}\int_{\Theta_j^*}\pi_j(\theta)\sum_{i=1}^{N_{j-1}}
    \bar{w}^{(i)}_{j-1}
q_j(\xi|\theta_{j-1}^{(i)})\min\left\{1,\frac{\pi_j(\xi)}{\pi_j(\theta_{j-1}^{(i)})}\right\}
\min\left\{1,\frac{\pi_j(\xi)\hat{\pi}_{j}^{N_{j-1}}(\theta)}{\pi_j(\theta)\hat{\pi}_{j}^{N_{j-1}}(\xi)}\right\}d\xi
d\theta\\
&\leq\int_{\Theta_j^*}\int_{\Theta_j^*}\pi_j(\theta)\sum_{i=1}^{N_{j-1}}
    \bar{w}^{(i)}_{j-1}
q_j(\xi|\theta_{j-1}^{(i)})\min\left\{1,\frac{\pi_j(\xi)}{\pi_j(\theta_{j-1}^{(i)})}\right\}d\xi
d\theta\\
&=\int_{\Theta_j^*}\pi_j(\theta)\sum_{i=1}^{N_{j-1}}
    \bar{w}^{(i)}_{j-1}a_j(\theta_{j-1}^{(i)})d\theta=\sum_{i=1}^{N_{j-1}}
    \bar{w}^{(i)}_{j-1}a_j(\theta_{j-1}^{(i)})
\end{split}
\end{equation}
To prove the lower bound (\ref{>}), we use (\ref{formaldefinition}) in the equation defining $\mathcal{\bar{A}}_j$:
\begin{equation}\label{2ndbound}
\begin{split}
\mathcal{\bar{A}}_j&=\int_{\Theta_j^*}\int_{\Theta_j^*}\pi_j(\theta)\hat{\pi}_{j}^{N_{j-1}}(\xi)
\min\left\{1,\frac{\pi_j(\xi)\hat{\pi}_{j}^{N_{j-1}}(\theta)}{\pi_j(\theta)\hat{\pi}_{j}^{N_{j-1}}(\xi)}\right\}
d\xi d\theta\\
&=\int_{\Theta_j^*}\int_{\Theta_j^*}
\pi_j(\theta)\hat{\pi}_{j}^{N_{j-1}}(\xi)I\left(
\frac{\pi_j(\xi)\hat{\pi}_{j}^{N_{j-1}}(\theta)}{\pi_j(\theta)\hat{\pi}_{j}^{N_{j-1}}(\xi)}\geq1\right)d\xi
d\theta\\
&+
\int_{\Theta_j^*}\int_{\Theta_j^*}
\pi_j(\theta)\hat{\pi}_{j}^{N_{j-1}}(\xi)I\left(
\frac{\pi_j(\theta)\hat{\pi}_{j}^{N_{j-1}}(\xi)}{\pi_j(\xi)\hat{\pi}_{j}^{N_{j-1}}(\theta)}\geq1\right)\frac{\pi_j(\xi)\hat{\pi}_{j}^{N_{j-1}}(\theta)}{\pi_j(\theta)\hat{\pi}_{j}^{N_{j-1}}(\xi)}d\xi
d\theta\\
&=2\int_{\Theta_j^*}\int_{\Theta_j^*}
\pi_j(\theta)\hat{\pi}_{j}^{N_{j-1}}(\xi)I\left(
\frac{\pi_j(\xi)\hat{\pi}_{j}^{N_{j-1}}(\theta)}{\pi_j(\theta)\hat{\pi}_{j}^{N_{j-1}}(\xi)}\geq1\right)d\xi
d\theta\\
&\geq2\int_{\Theta_j^*}\int_{\Theta_j^*}
\pi_j(\theta)\frac{\pi_j(\xi)}{M}I\left(
\frac{\pi_j(\xi)}{\hat{\pi}_{j}^{N_{j-1}}(\xi)}
\geq\frac{\pi_j(\theta)}{\hat{\pi}_{j}^{N_{j-1}}(\theta)}\right)d\xi
d\theta\\
&=\frac{2}{M}P\left(\frac{\pi_j(\xi)}{\hat{\pi}_{j}^{N_{j-1}}(\xi)}
\geq\frac{\pi_j(\theta)}{\hat{\pi}_{j}^{N_{j-1}}(\theta)}\right)=\frac{1}{M},
\end{split}
\end{equation}
where the last probability is equal to $1/2$, because $\theta$ and $\xi$ are
i.i.d. according to $\pi_j(\cdot)$, and hence the result.
\end{proof}
\begin{remark} The AIMS algorithm at annealing level $j$ has two accept/reject steps:
one is for the local candidate $\xi_l$ (step 1c) and another is for the
global candidate $\xi_g$ (step 2). The right-hand side of (\ref{<}) is nothing else but the local acceptance rate, i.e.
expected probability of generating a proper local candidate state $\xi_l\notin\{\theta_{j-1}^{(1)},\ldots,\theta_{j-1}^{(N_{j-1})}\}$.
Basically, (\ref{<}) says that the global acceptance rate $\bar{\mathcal{A}}_j$ can never exceed the local acceptance rate.
In fact, it can be deduced directly from the description of the algorithm, since
if the local candidate $\xi_l$ is rejected, then the global candidate $\xi_g$
is automatically rejected and we have a repeated sample
$\theta_j^{(i+1)}=\theta_j^{(i)}$.\label{globallocalAR}
\end{remark}

\section{Illustrative Examples}\label{examples-section}
In this section we illustrate the use of AIMS
with three examples: 1) mixture of ten Gaussian distributions in two dimensions (a multi-modal case);
2) sum of two multivariate Gaussian distributions in higher dimensions; and 3) Bayesian updating of a neural network model.

\subsection{Multi-modal mixture of Gaussians in 2D}\label{example-multimodal}
To demonstrate the efficiency of AIMS for sampling from multi-modal distributions,
consider simulation from a truncated two-dimensional mixture of $M$ Gaussian densities:
\begin{equation}\label{mixture}
    \pi(\theta)\propto \pi_0(\theta)\cdot L(\theta)=\mathcal{U}_{[0,a]\times[0,a]}(\theta)\cdot\sum_{i=1}^M w_i \mathcal{N}(\theta|\mu_i,\sigma^2\mathbb{I}_2),
\end{equation}
where $\mathcal{U}_{[0,a]\times[0,a]}(\cdot)$ denotes the uniform distribution on the square $[0,a]\times[0,a]$.
In this example, $a=10$, $M=10$, $\sigma=0.1$, $w_1=\ldots=w_{10}=0.1$, and the mean vectors $\mu_1,\ldots,\mu_{10}$ are
drawn uniformly from the square $[0,10]\times[0,10]$. Because of our interest in Bayesian updating, we refer to $\pi(\cdot)$
in (\ref{mixture}) as a posterior distribution.

Figure \ref{posterior samples}(a) displays the scatterplot of $10^3$ posterior samples obtained from AIMS.
Notice there are two clusters of samples that overlap significantly near $\theta=(4,4)$ that reflect two closely spaced
Gaussian densities but the other $8$ clusters are widely spaced.
The parameters of the algorithm were chosen as follows:
sample size $N=10^3$ per annealing level; the threshold for the ESS $\gamma=1/2$; the local proposal density
$q_j(\cdot|\xi)=\mathcal{N}(\cdot|\xi,c^2\mathbb{I}_2)$, with $c=0.2$. The trajectory of the corresponding posterior Markov chain,
i.e. the chain generated at the last annealing level with stationary distribution $\pi(\cdot)$,
is shown in Figure~\ref{posterior samples}(b). Black crosses $\times$ represent the mean vectors $\mu_1,\ldots,\mu_{10}$.
As expected, the chain does not exhibit a local random walk behavior
and it moves freely between well-separated modes of the posterior distribution.

The described implementation of AIMS leads to a total number of $m=6$ intermediate distributions in the annealing scheme.
Figure \ref{beta} shows how annealing parameter $\beta_j$ changes as a function of $j$ for $50$ independent runs of the algorithm.
It is found that in all considered examples, $\beta_j$ grows exponentially with $j$.

Let us now compare the performance of AIMS with the Random Walk Metropolis-Hastings algorithm.
For a fair comparison, the Metropolis-Hastings algorithm was implemented as follows. First, a sample of $N_0=10^3$ points
$\theta_0^{(1)},\ldots,\theta_0^{(N_0)}$ was drawn from the prior distribution $\pi_0(\cdot)=\mathcal{U}_{[0,a]\times[0,a]}(\cdot)$ and
the corresponding values of the likelihood function $L(\theta)=\sum_{i=1}^M w_i \mathcal{N}(\theta|\mu_i,\sigma^2\mathbb{I}_2)$ were calculated,
$L_i=L(\theta_0^{(i)})$. Then, starting from the point with the largest likelihood, $\theta^{(1)}=\theta^{(k)}_0$, $k=\arg\max L_i$,
a Markov chain $\theta^{(1)},\ldots,\theta^{(N)}$, with stationary distribution $\pi(\cdot)$ was generated
using the Metropolis-Hastings algorithm. The proposal distribution used was $q(\cdot|\xi)=N(\cdot|\xi,c^2\mathbb{I}_2)$ with $c=0.2$,
and the length of the chain was $N=5\cdot10^3$. Thus, the total number of samples used in both AIMS and RWMH was $N_{t}=6\cdot10^3$.
The scatterplot of posterior samples obtained from RWMH and the trajectory of the corresponding Markov chain are show in
Figures~\ref{posterior samples}(c) and \ref{posterior samples}(d), respectively.
While the AIMS algorithm successfully sampled all $10$ modes with the approximately correct proportion of total samples,
RWHM completely missed $7$ modes.

Suppose that we are interested in estimating the posterior mean vector, $\mu^{\pi}=(\mu^{\pi}_1,\mu^{\pi}_2)$, and
the components $(\sigma_1^{\pi})^2, (\sigma_2^{\pi})^2, \sigma_{12}^{\pi}$ of the posterior covariance matrix $\Sigma^{\pi}$.
Their true values are given in Table \ref{true values} along with the AIMS estimates in terms of their means and coefficients of variation
averaged over 50 independent simulations, all based on $10^3$ posterior samples.

Figure \ref{Mean_square_error} displays the mean square error (MSE) of the
AIMS estimator for the posterior mean and covariance matrix for different values of the scaling factor $c$.
The MSE was estimated based on $50$ independent runs of the algorithm.
An interesting observation is that the MSE as a function of $c$ is nearly flat around the optimal, $c_{\mathrm{opt}}\approx0.15$,
i.e. the one that minimizes the MSE.

\subsection{Mixture of two higher-dimensional Gaussians}\label{multGauss}

To demonstrate the efficiency of AIMS for higher dimensionality,
consider simulation from a truncated  sum of two multivariate Gaussian densities:
\begin{equation}\label{sum of two Gaussian}
     \pi^d(\theta)\propto \pi_0^d(\theta)\cdot L^d(\theta)=\mathcal{U}_{[-a,a]^d}(\theta)\cdot
     \left(\mathcal{N}(\theta|\mu_1,\sigma^2\mathbb{I}_d) + \mathcal{N}(\theta|\mu_2,\sigma^2\mathbb{I}_d)\right),
\end{equation}
where $a=2$, $\mu_1=(0.5,\ldots,0.5)$, $\mu_2=(-0.5,\ldots,-0.5)$, and $\sigma=0.5$. Thus,
$\pi^d(\cdot)$ is a bimodal distribution on a $d$-dimensional cube $[-a,a]^d$. Suppose that a quantity of
interest is the function $h: [-a,a]^d\rightarrow[-a,a]$ that gives the largest component of $\theta=(\theta_1,\ldots,\theta_d)\in[-a,a]^d:$
\begin{equation}\label{h}
    h(\theta)=\max\{\theta_1,\ldots,\theta_d\}
\end{equation}
and we want to estimate its expectation with respect to $\pi^d(\cdot)$ using posterior samples $\theta^{(1)},\ldots,\theta^{(N)}\sim\pi^d(\cdot)$
as follows:
\begin{equation}\label{estimator}
    \bar{h}=\mathbb{E}_{\pi^d}[h]\approx \hat{h}_N=\frac{1}{N}\sum_{i=1}^N h(\theta^{(i)})
\end{equation}
This example is taken from \cite{Ching}, where the Transitional Markov chain Monte Carlo method (TMCMC)
for sampling from posterior densities was introduced.

Here, we consider five cases: $d=2,4,6,10,$ and $20$.  The performance of TMCMC was examined for only the first three cases in \cite{Ching}.
The last two cases are higher dimensional, and, therefore, more challenging.

%As in the previous example, $\gamma=1/2$ was used.

The details of implementation and simulation results from 50 independent runs are summarized in Table \ref{AMIS vs TMCMC}.
First of all, observe that AIMS outperforms TMCMC, when $d=2,4,6$. Both methods are capable of generating samples
from both modes of the posterior; however, the probabilities of the modes (each is $1/2$ in this example) are found more accurately by AIMS.

\begin{remark} In addition to the first three cases, five other scenarios with different probabilities of modes and different values of $\sigma$
were examined in \cite{Ching}. It is found that AIMS outperforms TMCMC in all these cases too.
\end{remark}

Results presented in Table \ref{AMIS vs TMCMC} help to shed some light on the properties of the optimal scaling parameter $c_{\mathrm{opt}}$
for the proposal density $q_j(\cdot|\xi)=\mathcal{N}(\cdot|\xi,c^2\mathbb{I}_d)$. It appears $c_{\mathrm{opt}}$ depends not only on the
dimension $d$, which is expected, but also on $N$, the number of samples used per each annealing level.
The latter dependence is explained by the fact that the global proposal distribution
$\hat{\pi}_j^{N}(\cdot)$ for the AIMS Markov chain depends both on $N$ and $c$: $\hat{\pi}_j^{N}(\cdot)$
is a weighted sum of $N$ RWMH transition kernels with Gaussian proposal distributions, whose spread is controlled by $c$.
When $N$ is fixed, $c_{\mathrm{opt}}$ is a monotonically increasing function of $d$,
since in higher dimensions, for optimal local exploration of the neighborhoods of $\theta_{j-1}^{(1)},\ldots,\theta_{j-1}^{(N)}$,
we have to be able to make larger local jumps from $\theta_{j-1}^{(k)}$ to $\xi_l$.
When $d$ is fixed, $c_{\mathrm{opt}}$ is a monotonically decreasing function of $N$, since the more samples
$\theta_{j-1}^{(1)},\ldots,\theta_{j-1}^{(N)}$ that have been generated at the previous level, the more we can focus on
local exploration of their neighborhoods
without worrying too much about regions that lie far away. If we think of the support of
$q_j(\cdot|\theta_{j-1}^{(k)})=\mathcal{N}(\cdot|\theta_{j-1}^{(k)},c^2\mathbb{I}_d)$
as lying mostly in a $d$-dimensional ball of radius $c$ centered at $\theta_{j-1}^{(k)}$,
then we can explain the dependence of $c_{\mathrm{opt}}$ on $N$ as follows:
the more $d$-dimensional balls of radius $c$ we have, the smaller $c$ we can use for covering the sample space.

It is interesting to look at how the local and global acceptance rates (see Remark \ref{globallocalAR})
depend on the scaling parameter $c$. Figures \ref{fig1}, \ref{fig2}, and \ref{fig3} display these acceptance rates along with the coefficient of
variation $\delta$ of the AIMS estimator for the first three cases: $d=2,4$ and $6$, based on 50 independent runs.
As expected, the global acceptance rate is always smaller than the local acceptance rate,
and the minimum value of $\delta$ corresponds to the maximum value of the global acceptance rate.
Observe also that the peak of the global acceptance rate slides to the left, when $j$ increases.
This suggests that it is more efficient to use smaller values of $c$ at higher annealing levels.
Indeed, it is natural to expect that $c_{j}^{\mathrm{opt}}>c_{j+1}^{\mathrm{opt}}$, since the intermediate
distribution $\pi_{j+1}(\cdot)$ is more concentrated  than $\pi_{j}(\cdot)$.

Finally, we draw attention to Case $4$ in Table 2 where $d=10$ with $N=10^3$ and $N=2\cdot10^3$ samples per annealing level.
Usually for Monte Carlo based methods, the coefficient of variation $\delta$ of the estimator is proportional to $1/\sqrt{N_t}$, where $N_t$ is the total
number of samples. Thus, the doubling of sample size will result in the reduction of $\delta$ by the factor of $1/\sqrt{2}\approx0.71$.
For AIMS, however, the decrease of $\delta$ is more significant: from $\delta=26.7\%$ to $\delta=12.2\%$, i.e. approximately by the factor
of $0.46$. This is because, as explained in Subsection \ref{impl issues},
the increase of $N$ affects not only the total sample size, but also improves the global proposal distribution $\hat{\pi}_j^N(\cdot)$.
This improvement of $\hat{\pi}_j^N(\cdot)$ results in the generation of less correlated samples at each annealing level and, therefore,
leads to an additional reduction of the coefficient of variation $\delta$.

\subsection{Bayesian updating of a neural network}\label{BNN}

To illustrate the use of AIMS for Bayesian updating, consider its application to a feed-forward neural network model, one of
the most popular and most widely used models for function approximation. The goal is to approximate a (potentially highly nonlinear) function
$f:X\rightarrow \mathbb{R}$, where $X\subset\mathbb{R}^p$ is a compact set, based on a finite number of measurements $y_i=f(x_i)$, %+\varepsilon_i$,
$i=1,\ldots,n$, by using a finite sum of the form
\begin{equation}\label{FFNN1}
\hat{f}(x,\theta)=\sum_{j=1}^M\alpha_j\Psi(\langle x,\beta_j\rangle+\gamma_j)
\end{equation}
where $\theta$ denotes the model parameters $\alpha_j,\gamma_j\in\mathbb{R}$ and $\beta_j\in\mathbb{R}^p$,
$\langle\cdot,\cdot\rangle$ is the standard scalar product in $\mathbb{R}^p$, and $\Psi$ is a sigmoidal function, the typical
choice being either the logistic function or the $\tanh$ function that is used in this example:
\begin{equation}\label{FFNN2}
    \Psi(z)=\frac{e^z-e^{-z}}{e^z+e^{-z}}.
\end{equation}
Model (\ref{FFNN1}) is called a feed-forward neural network (FFNN) with activation function (\ref{FFNN2}), $p$ input units, one hidden layer with
$M$ hidden units, and one output unit. The parameters $\beta_j$ and $\alpha_j$ are called the connection weights from the input units to the hidden unit $j$
and the connection weights from the hidden unit $j$ to the output unit, respectively. The term $\gamma_j$ is a designated bias of the
hidden unit $j$ and it can be viewed as a connection weight from an additional constant unit input. % In this example, we omit the bias terms
%by considering them as components of the corresponding vectors $\beta_j$.
 Schematically, the FFNN model is shown in Figure~\ref{FFNNmodel}.

The rationale behind the FFNN approximation method follows from the universal approximation property of FFNN models \cite{Cybenko,Hornik};
that is, a FFNN with sufficient number of hidden units and properly adjusted connection weights can approximate most functions arbitrarily well.
More precisely, finite sums (\ref{FFNN1}) over all positive integers $M$ are dense in the set of real continuous functions on the $p$-dimensional unit cube.

Let $\mathcal{A}$ denote the FFNN architecture, i.e. the input-output model (\ref{FFNN1}) together with
information about the type of activation function $\Psi$, number of input units $p$, and number of hidden units $M$.
In this example, we use $p=1$, $M=2$, and $\Psi$ is given by (\ref{FFNN2}), so the model parameters
$\theta=(\alpha_1, \alpha_2, \beta_{1}, \beta_{2}, \gamma_{1}, \gamma_{2})\in\Theta=\mathbb{R}^6$.
%Here, $\gamma_1$ and $\gamma_2$ are $\beta_{11}$ and $\beta_{21}$, %respectively.

\textit{Deterministic model} $\mathcal{A}$ of function $f$ given by $\hat{f}(x,\theta)$ in (\ref{FFNN1}) can be used to construct
a \textit{Bayesian (stochastic) model} $\mathcal{M}$ of function $f$ by \textit{stochastic embedding} (see the details in \cite{Beck1,Beck2}).
Recall, that by definition, a Bayesian model $\mathcal{M}$ consists of two components:
\begin{enumerate}
  \item An input-output probability model $y\sim p(y|x,\theta,\mathcal{M})$, which is obtained by introducing the prediction-error
\begin{equation}\label{error}
\varepsilon=y-\hat{f}(x,\theta),%\sum_{j=1}^M\alpha_j\Psi(\langle x,\beta_j\rangle),
\end{equation}
which is the difference between the true output $y=f(x)$ and the deterministic model output
$\hat{f}(x,\theta)$. %$\sum_{j=1}^M\alpha_j\Psi(\langle x,\beta_j\rangle)$.
A probability model for $\varepsilon$ is introduced by using the Principle of Maximum Entropy \cite{Jaynes1,Jaynes2},
which states that the probability model should be selected to produce the most uncertainty subject
to constraints that we wish to impose (the selection of any other probability model would lead to an unjustified
reduction in the prediction uncertainty). In this example, we impose the following constraints:
$\mathbb{E}[\varepsilon]=0$ and $\mbox{var}[\varepsilon]=\sigma^2$ with $\varepsilon$ unbounded.
The maximum entropy PDF for the prediction-error is then
$\varepsilon\sim\mathcal{N}(0,\sigma^2)$. This leads to the following input-output probability model:
\begin{equation}\label{iomodel}
    p(y|x,\theta,\mathcal{M})=\mathcal{N}\left(\left.y \hspace{1mm} \right|\hat{f}(x,\theta), \sigma^2\right)
\end{equation}
Here, the prediction-error variance $\sigma^2$ is included in the set of model parameters where, for convenience, we define $\theta_7=\log\sigma^{-2}$,
so the parameter space is now $\Theta=\mathbb{R}^7$.

  \item A prior PDF $\pi_0(\theta|\mathcal{M})$ over the parameter space which is chosen to quantify
  the initial relative plausibility of each value of $\theta$ in $\Theta$.
  %After the parameter space has been augmented with $\theta_7$,
  %it becomes $7$-dimensional, $\Theta=\mathbb{R}^7$.
  In this example, the prior distributions are assumed to be:
  \begin{equation}\label{priors}
   %\alpha_j\sim \mathcal{N}(0,\sigma_{\alpha}^2), \hspace{3mm} \beta_j\sim\mathcal{N}(0,\sigma^2_{\beta}\mathbb{I}_2), %\hspace{3mm} j=1,2,
   \alpha_j\sim \mathcal{N}(0,\sigma_{\alpha}^2), \hspace{3mm} \beta_j\sim\mathcal{N}(0,\sigma^2_{\beta}), %\hspace{3mm} j=1,2,
   \hspace{3mm} \gamma_j\sim\mathcal{N}(0,\sigma^2_{\gamma}), \hspace{3mm}\theta_7=\log\sigma^{-2}\sim\mathcal{N}(0,\sigma_{\theta_7}^2),
  \end{equation}
with $\sigma_{\alpha}=\sigma_{\beta}=\sigma_{\gamma}=\sigma_{\theta_7}=5$. Thus, the prior PDF in our case is
\begin{equation}\label{priorPDFex}
\pi_0(\theta|\mathcal{M})=\mathcal{N}(\theta_7|0,\sigma_{\theta_7}^2)\prod_{j=1}^M\mathcal{N}(\alpha_j|0,\sigma_{\alpha}^2)
\mathcal{N}(\beta_{j}|0,\sigma^2_{\beta})\mathcal{N}(\gamma_{j}|0,\sigma^2_{\gamma}).
\end{equation}
\end{enumerate}

Let $\mathcal{D}$ denote the training data, $\mathcal{D}=\{(x_1,y_1),\ldots,(x_n,y_n)\}$, treated as independent samples, then the likelihood function which
expresses the probability of getting data $\mathcal{D}$ based on the probability model (\ref{iomodel}) is given by
\begin{equation}\label{likelihoodfunction}
    L(\theta)=p(\mathcal{D}|\theta,\mathcal{M})=\prod_{i=1}^n p(y_i|x_i,\theta,\mathcal{M})
\end{equation}
In this example, data are synthetically generated from (\ref{iomodel}) with $\alpha_1=5$, $\alpha_2=-5$, $\beta_1=-1$, $\beta_2=-3$,
$\gamma_1=5$, $\gamma_2=2$, $\sigma=0.1$, and the input $x_i=i/10$, for $i=1,\ldots,n=100$.

Finally, using Bayes’ theorem, we can write the posterior PDF $\pi(\theta|\mathcal{D},\mathcal{M})$ for the uncertain model parameters:
\begin{equation}\label{posteriorPDFBNN}
\begin{split}
\pi(\theta|\mathcal{D},\mathcal{M})&\propto \pi_0(\theta|\mathcal{M})\cdot L(\theta)\\
&=\mathcal{N}(\theta_7|0,\sigma_{\theta_7}^2)\prod_{j=1}^M\mathcal{N}(\alpha_j|0,\sigma_{\alpha}^2)
\mathcal{N}(\beta_{j}|0,\sigma^2_{\beta})\mathcal{N}(\gamma_{j}|0,\sigma^2_{\gamma})
\cdot \prod_{i=1}^n p(y_i|x_i,\theta,\mathcal{M})
\end{split}
\end{equation}

Under the Bayesian framework, the mean prediction of $y=f(x)$ from observable $x$
 can be obtained by integrating out the nuisance parameters:
\begin{equation}\label{prediction}
    %\bar{\hat{f}}_{\pi}(x)=
    \mathbb{E}_{\pi}[y|x,\mathcal{D},\mathcal{M}]=\int_\Theta \hat{f}(x,\theta)\pi(\theta|\mathcal{D},\mathcal{M})d\theta
\end{equation}

To demonstrate the efficiency of AIMS for the mean prediction problem, we use it to sample from the posterior PDF (\ref{posteriorPDFBNN})
and use Monte Carlo simulation in (\ref{prediction}).
The parameters of the AIMS algorithm are chosen as follows: sample size $N=3\times10^3$ per annealing level;
the threshold for the ESS $\gamma=1/2$; the proposal density
$q_j(\cdot|\xi)=\mathcal{N}(\cdot|\xi,c^2\mathbb{I}_7)$, with $c=0.5$. This implementation of AIMS leads to a total number
of $m=10$ intermediate distributions in the annealing scheme. The obtained posterior samples $\theta_m^{(1)},\ldots,\theta_m^{(1)}$
are then used to approximate the integral on the right-hand side of (\ref{prediction}):
\begin{equation}\label{approximatedprediction}
\int_\Theta \hat{f}(x,\theta)\pi(\theta|\mathcal{D},\mathcal{M})d\theta\approx \frac{1}{N}\sum_{i=1}^N \hat{f}(x,\theta_m^{(i)})\overset{\underset{\mathrm{def}}{}}{=}
\bar{\hat{f}}_m(x)
\end{equation}
The true function $y=f(x)$  as well as its AIMS approximation $\bar{\hat{f}}_m(x)$ are shown in Figure~\ref{FFNNexample}.
A few ``intermediate approximations'' $\bar{\hat{f}}_j(x)$, which are based on $\theta_j^{(1)},\ldots,\theta_j^{(1)}\sim\pi_j$, are plotted to show how $\bar{\hat{f}}_j(x)$ approaches $f(x)$ when $j\rightarrow m$.
To visualize the uncertainty for the AIMS approximation, we plot its $5$th and $95$th percentiles in Figure~\ref{FFNNexample2}.

\section{Concluding Remarks}\label{finish}
In this paper, a new scheme for sampling from posterior distributions, called Asymptotically Independent Markov Sampling (AIMS), is introduced.
The algorithm is based on three well-established and widely-used stochastic simulation methods: importance sampling, MCMC, and simulated annealing.
The key idea behind AIMS is to use $N$
samples drawn from $\pi_{j-1}(\cdot)$ as an importance sampling density to
construct an approximation $\hat{\pi}_{j}^N(\cdot)$ of
$\pi_{j}(\cdot)$, where $\pi_0(\cdot),\ldots,\pi_m(\cdot)$ is a sequence of intermediate distributions interpolating between the prior $\pi_0(\cdot)$
and posterior $\pi(\cdot)=\pi_m(\cdot)$. This approximation is then employed as the independent
proposal distribution for sampling from $\pi_{j}(\cdot)$ by the independent Metropolis-Hastings algorithm.
When $N\rightarrow\infty$, the AIMS
sampler generates independent draws from the target distribution, hence the name of the
algorithm.

Important ergodic properties of AIMS are derived. In particular, it is shown that, under certain conditions (that are often fulfilled in practice),
the AIMS algorithm produces a uniformly ergodic Markov chain. The choice of the free parameters of the algorithm is discussed
and recommendations  are provide for their values, both theoretically and heuristically based.
The efficiency of AIMS is demonstrated with three examples, which include
both multi-modal and higher-dimensional target posterior distributions.

\section*{Acknowledgements}\label{Acknowledgements}
This work was supported by the National Science Foundation under award
number EAR-0941374 to the California Institute of Technology. This support is gratefully acknowledged.
Any opinions, findings, and conclusions
or recommendations expressed in this material are
those of the authors and do not necessarily reflect
those of the National Science Foundation.
%\newpage

\newpage

\begin{table}%[h]
  \begin{center}
    \begin{tabular}{|l|c|c|c|c|c|}
    \hline
    Parameter & $\mu^{\pi}_1$ & $\mu^{\pi}_2$ & $(\sigma_1^{\pi})^2$ & $(\sigma_2^{\pi})^2$ & $\sigma_{12}^{\pi}$ \\
    \hline
    \hline
  True value & 5.23 & 5.75 & 4.51 & 3.37 & -1.30\\
  \hline
  AIMS mean & 5.20 & 5.73 & 4.56 & 3.32 & -1.25\\
  \hline
  AIMS cov & 2.4\% & 2.0\% & 8.2\% & 8.2\% & 27.7\%\\
  \hline
    \end{tabular}
  \end{center}
\caption{\footnotesize True values of the posterior parameters and the AIMS estimates in terms of
their means and coefficients of variation averaged over 50 simulations [Example \ref{example-multimodal}].}\label{true values}
\end{table}

\begin{table}%[h]
  \begin{center}
    \begin{tabular}{|c|c|c|c|c|c|c|c|c|}
    \hline
    Case & $d$ & $\bar{h}$ & TMCMC: $\hat{h}_N, (\delta$) & AIMS: $\hat{h}_N, (\delta$)& $N$ & $\gamma$ & $c_{\mathrm{opt}}$ & $\bar{m}$ \\
    \hline
     1 & 2 & 0.29 & 0.28 (12.3\%) & 0.29 (8.8\%) & $10^3$& 1/2 &$0.2$ &3\\
     2 & 4 & 0.51 & 0.54 (10.0\%) & 0.51 (6.9\%) & $10^3$& 1/2 &$0.4$ &4\\
     3 & 6 & 0.64 & 0.65 (15.7\%) & 0.64 (10.4\%) &$10^3$& 1/2 &$0.6$ &4.95\\
     \hline
     4 & 10 & 0.76   & ---  &0.76 (26.7\%) &$10^3$&1/2 &0.7& 5.84\\
       & 10 & 0.76   & ---  &0.76 (12.2\%) &$2\cdot10^3$&1/2 &0.6& 5.98\\
     5 & 20 & 0.92   & --- &0.95 (42.1\%)& $4\cdot10^3$ &1/2 & $0.5$ &  5.58\\
  \hline
    \end{tabular}
  \end{center}
\caption{\footnotesize  Summary of the simulation results: $d$ is the dimension of the sample space;
$\bar{h}$ and $\hat{h}_N$ are the exact value of $\mathbb{E}_{\pi^d}[h]$ and its estimated value, respectively; $\delta$ in parentheses is
the corresponding coefficient of variation; $N$, $\gamma$, $c_{\mathrm{opt}}$, and $\bar{m}$ are the number of samples used per annealing level,
the threshold for the ESS, the (nearly) optimal value of the scaling parameter, and the average number of distributions in the annealing scheme,
respectively. The AIMS results are based on $50$
independent runs. The TMCMC results are taken from \cite{Ching} and are based on 50 independent runs [Example \ref{multGauss}].
}\label{AMIS vs TMCMC}
\end{table}

\begin{figure}\centering
\includegraphics[angle=0,scale=1]{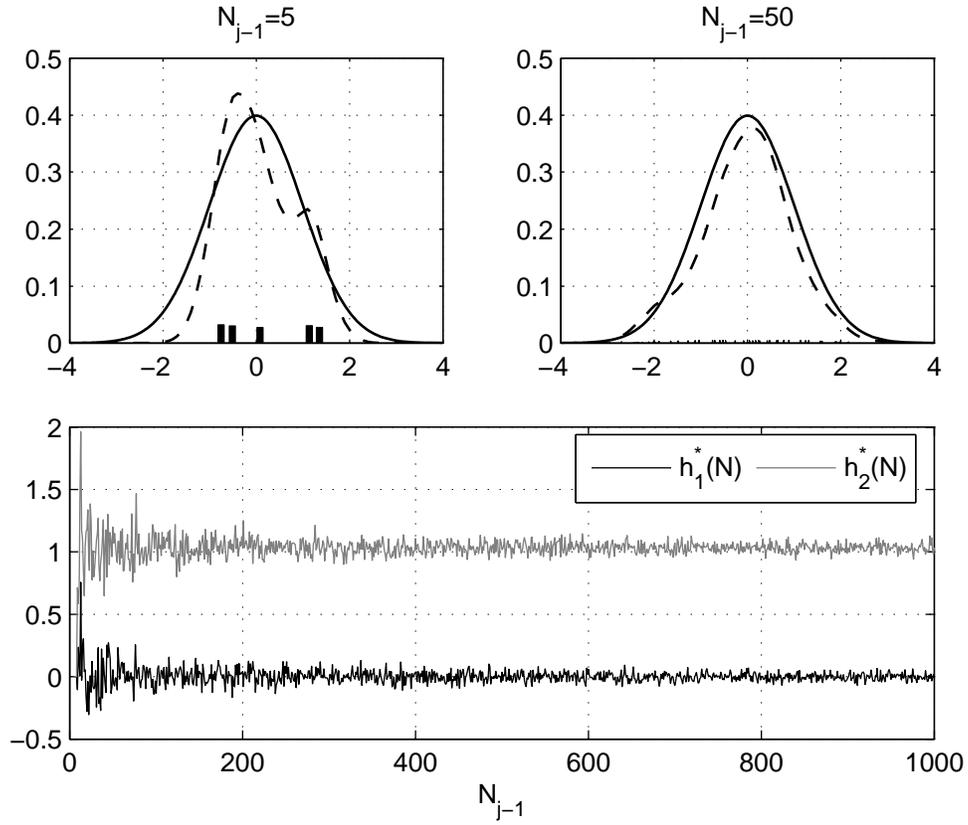}
\caption{\footnotesize The top panels show
the distribution $\pi_j(\cdot)$ (solid lines) and its approximation
$\hat{\pi}_j^{N_{j-1}}(\cdot)$, for $N_{j-1}=5$ (left) and
$N_{j-1}=50$ (right). Dashed lines and bars correspond to the continuous
and discrete parts of $\hat{\pi}_j^{N_{j-1}}(\cdot)$, respectively.
The bottom panel shows the convergence of
$h_1^*(N_{j-1})=\mathbb{E}_{\hat{\pi}_{j}^{N_{j-1}}}[h_1]$ and
$h_2^*(N_{j-1})=\mathbb{E}_{\hat{\pi}_{j}^{N_{j-1}}}[h_2]$ to the true
values, $0$ and $1$, respectively [Example 2.1].} \label{example1}
\end{figure}

\newpage

\begin{figure}\centering
\includegraphics[angle=0,scale=1]{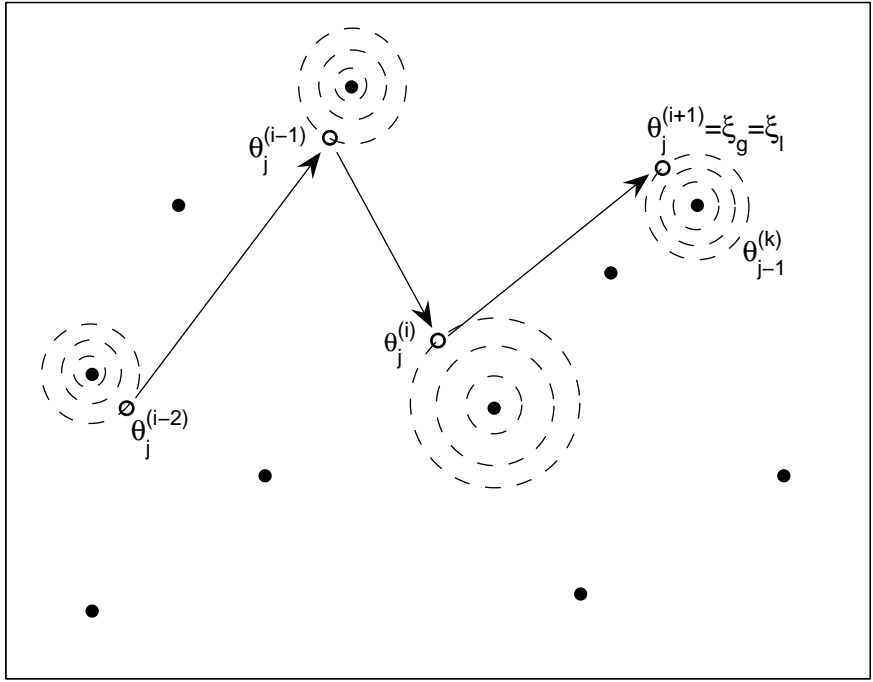}
\caption{\footnotesize AIMS at annealing level $j$: disks $\bullet$
and  circles $\circ$ represent
$\theta_{j-1}^{(1)},\ldots,\theta_{j-1}^{(N_{j-1})}$ and
$\theta_{j}^{(1)},\ldots,\theta_{j}^{(N_{j})}$, respectively;
concentric circles show the correspondence between
$\theta_{j-1}^{(k)}$ that has been chosen in step 1a and the
corresponding local candidate  $\xi_l\sim
q(\cdot|\theta_{j-1}^{(k)})$ that has been generated in step 1b. In
this schematic picture, all shown candidate states are accepted as
new states of the Markov chain.} \label{scheme}
\end{figure}

\newpage

\begin{figure}\centering
\includegraphics[angle=0,scale=1]{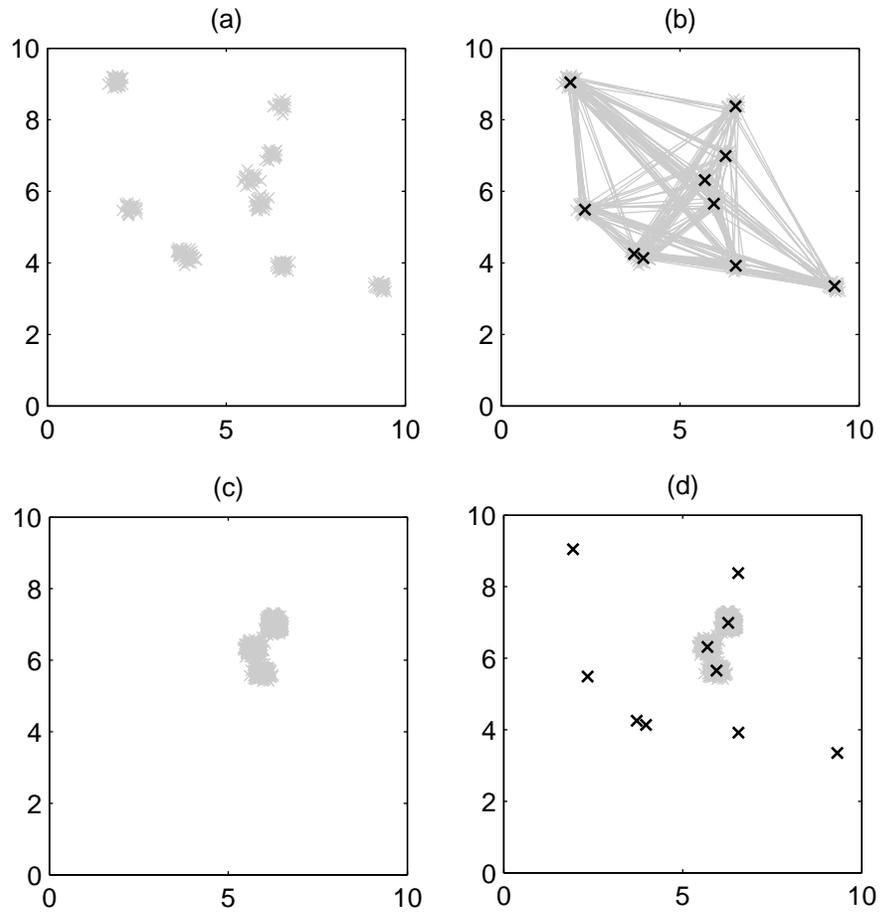}
\caption{\footnotesize (a) Scatterplots of $10^3$ posterior samples; (b) the trajectories of the corresponding posterior Markov chain
obtained from AIMS; and (c), (d) corresponding plots from RWMH.
Black crosses $\times$ represent the modes $\mu_1,\ldots,\mu_{10}$ of $\pi(\cdot)$ [Example 4.1].} \label{posterior samples}
\end{figure}

\newpage

\begin{figure}\centering
\includegraphics[angle=0,scale=1]{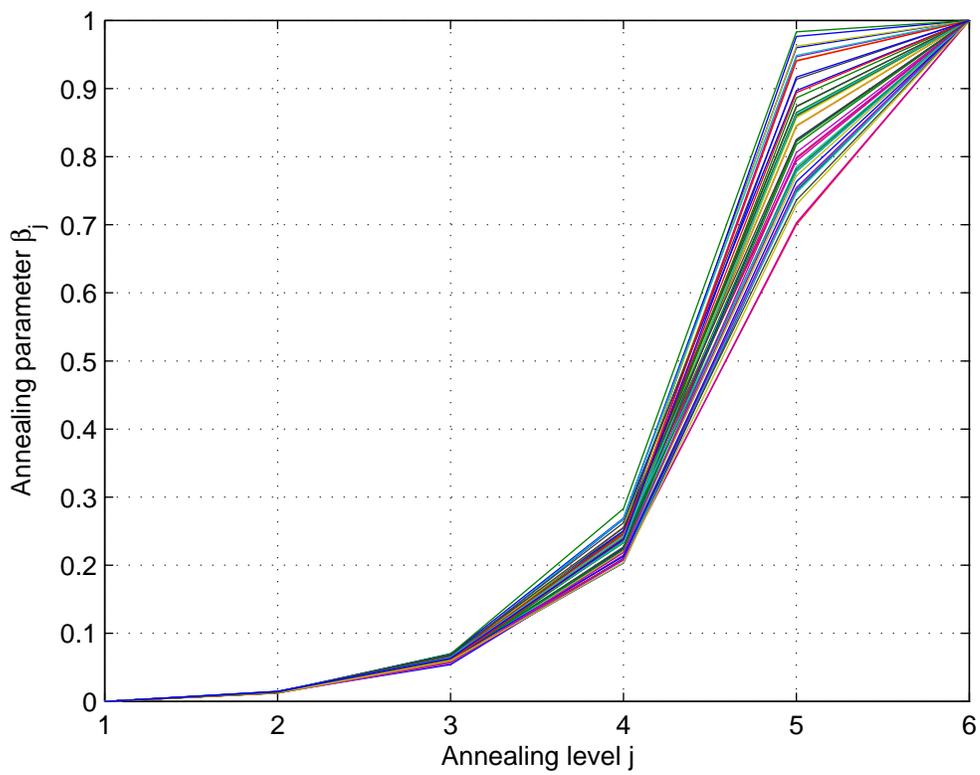}
\caption{\footnotesize Annealing parameter $\beta_j$ as a function of annealing level $j$ for $50$ independent runs of AIMS [Example 4.1].} \label{beta}
\end{figure}

\newpage

\begin{figure}\centering
\includegraphics[angle=0,scale=1]{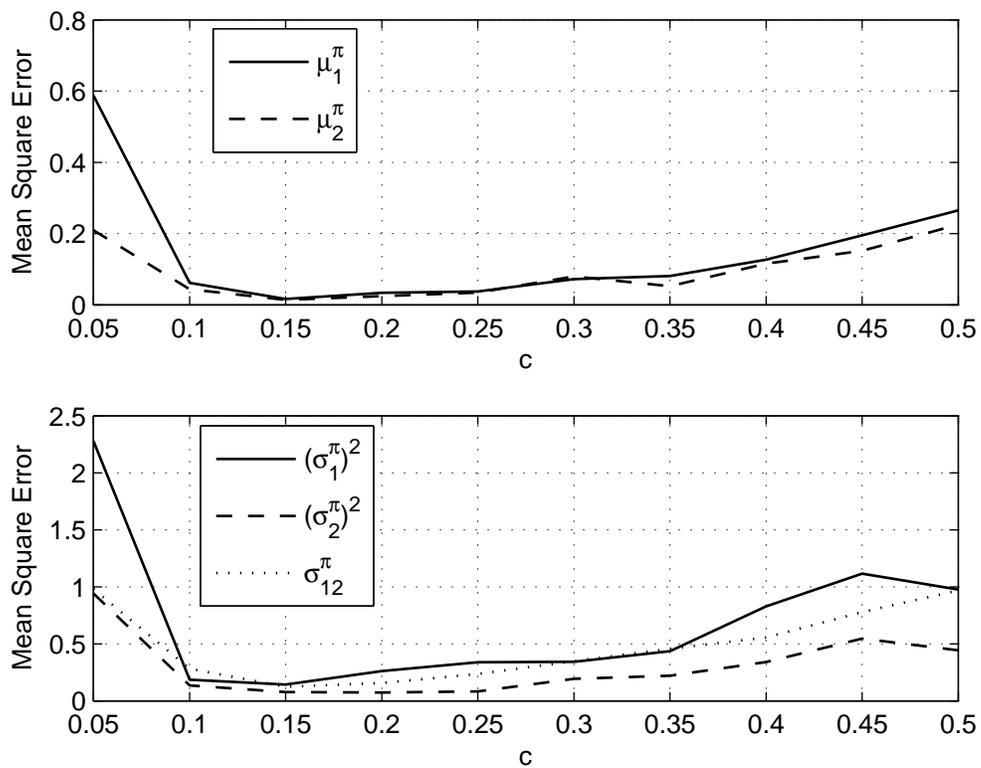}
\caption{\footnotesize Mean square error of the AIMS estimator for the mean and covariance matrix as a function of the scaling factor $c$
showing the optimal value is $c_{\mathrm{opt}}\approx0.15$ [Example \ref{example-multimodal}].} \label{Mean_square_error}
\end{figure}

\newpage

\begin{figure}\centering
\includegraphics[angle=0,scale=0.7]{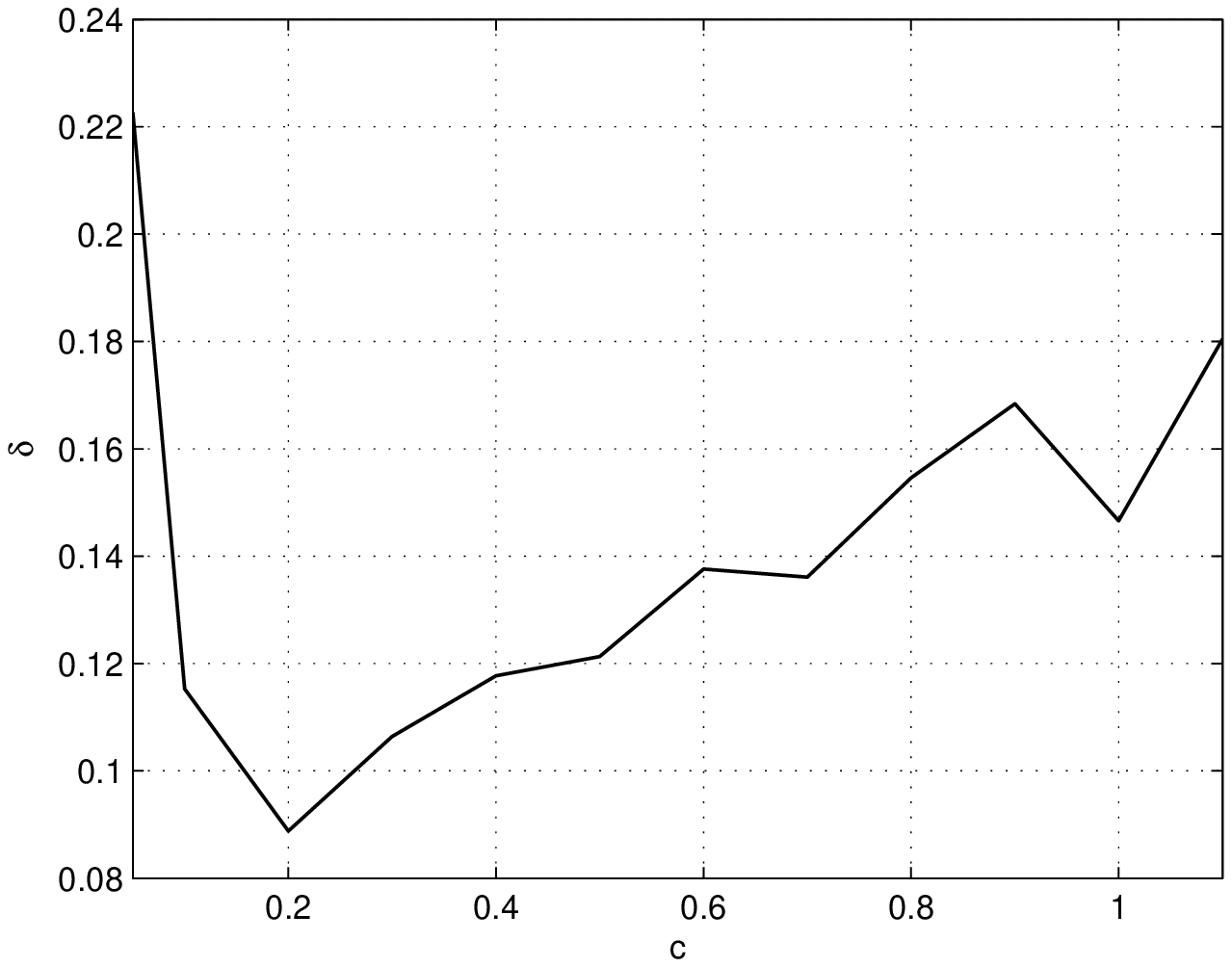}
\includegraphics[angle=0,scale=.7]{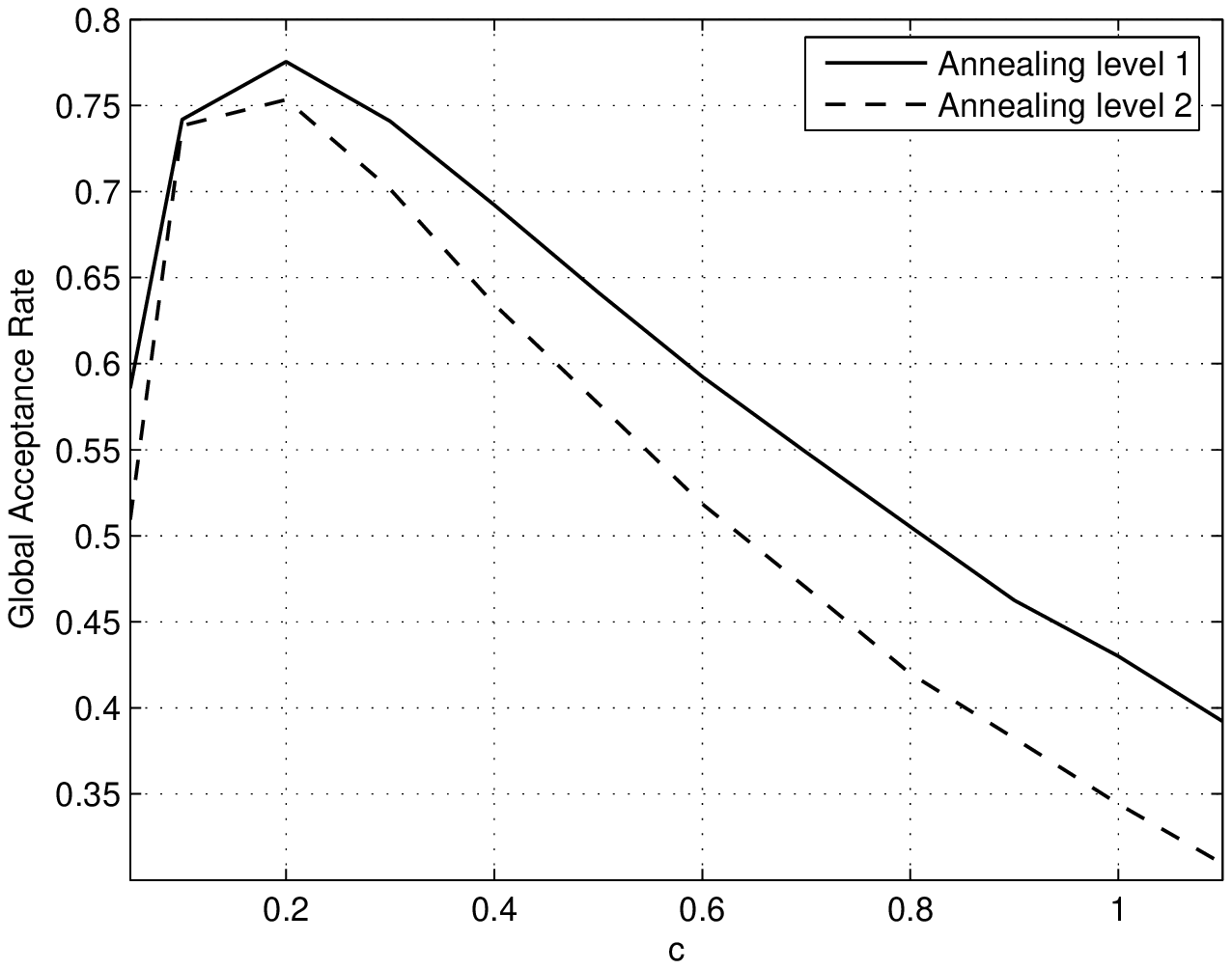}
\includegraphics[angle=0,scale=.7]{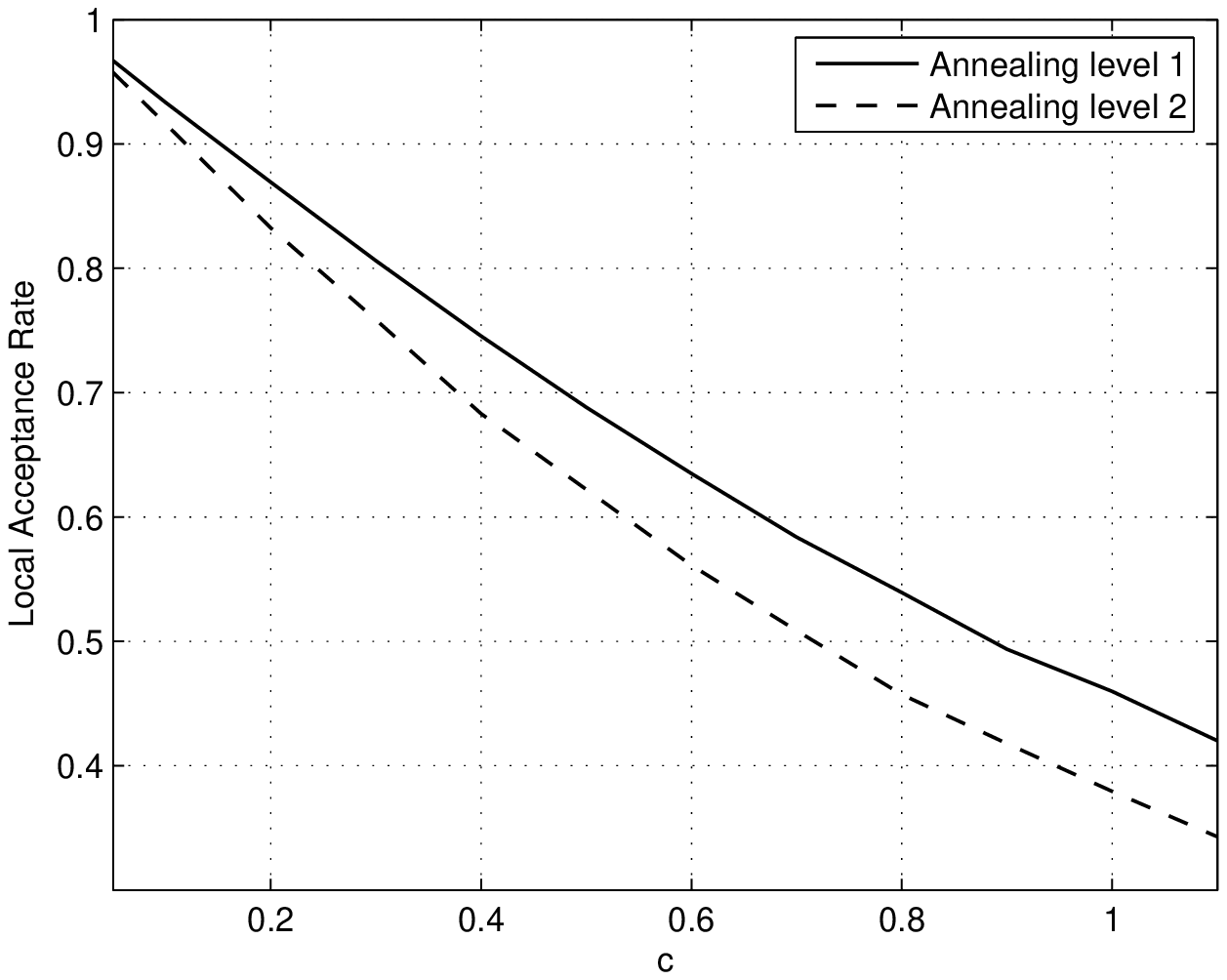}
\caption{\footnotesize Coefficient of variation $\delta$ of the AIMS estimate (top panel), global acceptance rate (middle panel), and
local acceptance rate (bottom panel) as  functions of $c$ for Case 1 ($d=2$) [Example \ref{multGauss}].}\label{fig1}
\end{figure}

\begin{figure}\centering
\includegraphics[angle=0,scale=0.7]{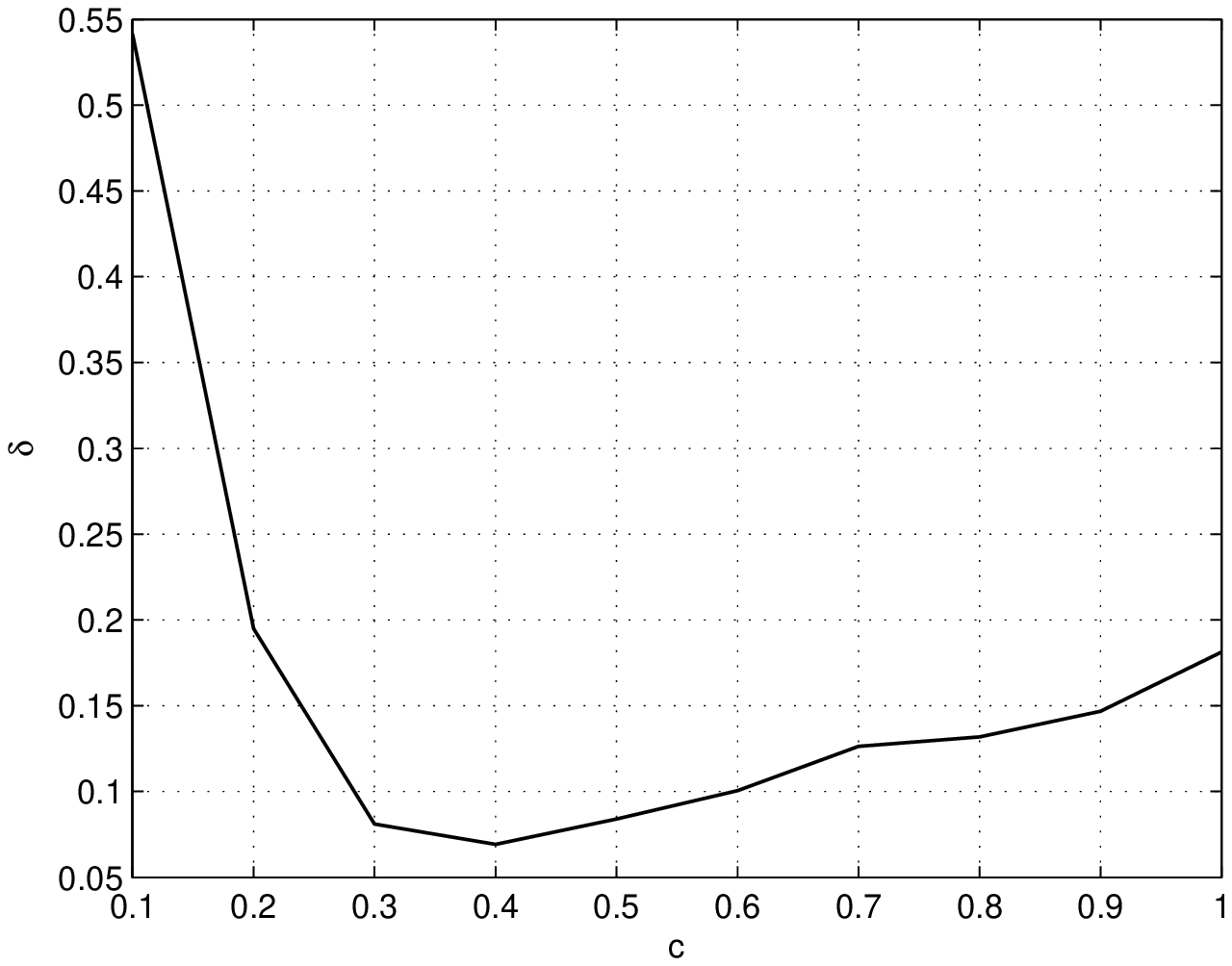}
\includegraphics[angle=0,scale=.7]{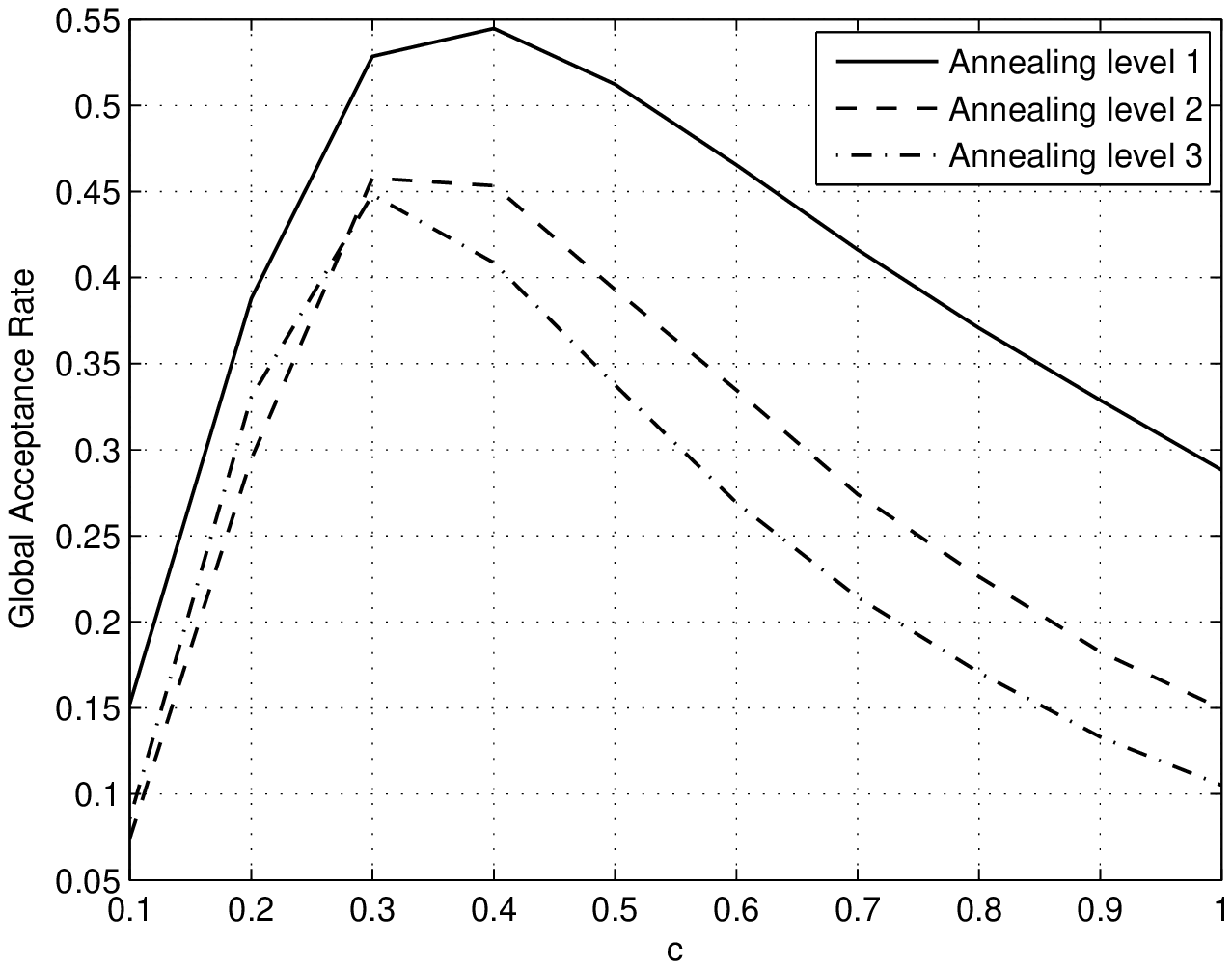}
\includegraphics[angle=0,scale=.7]{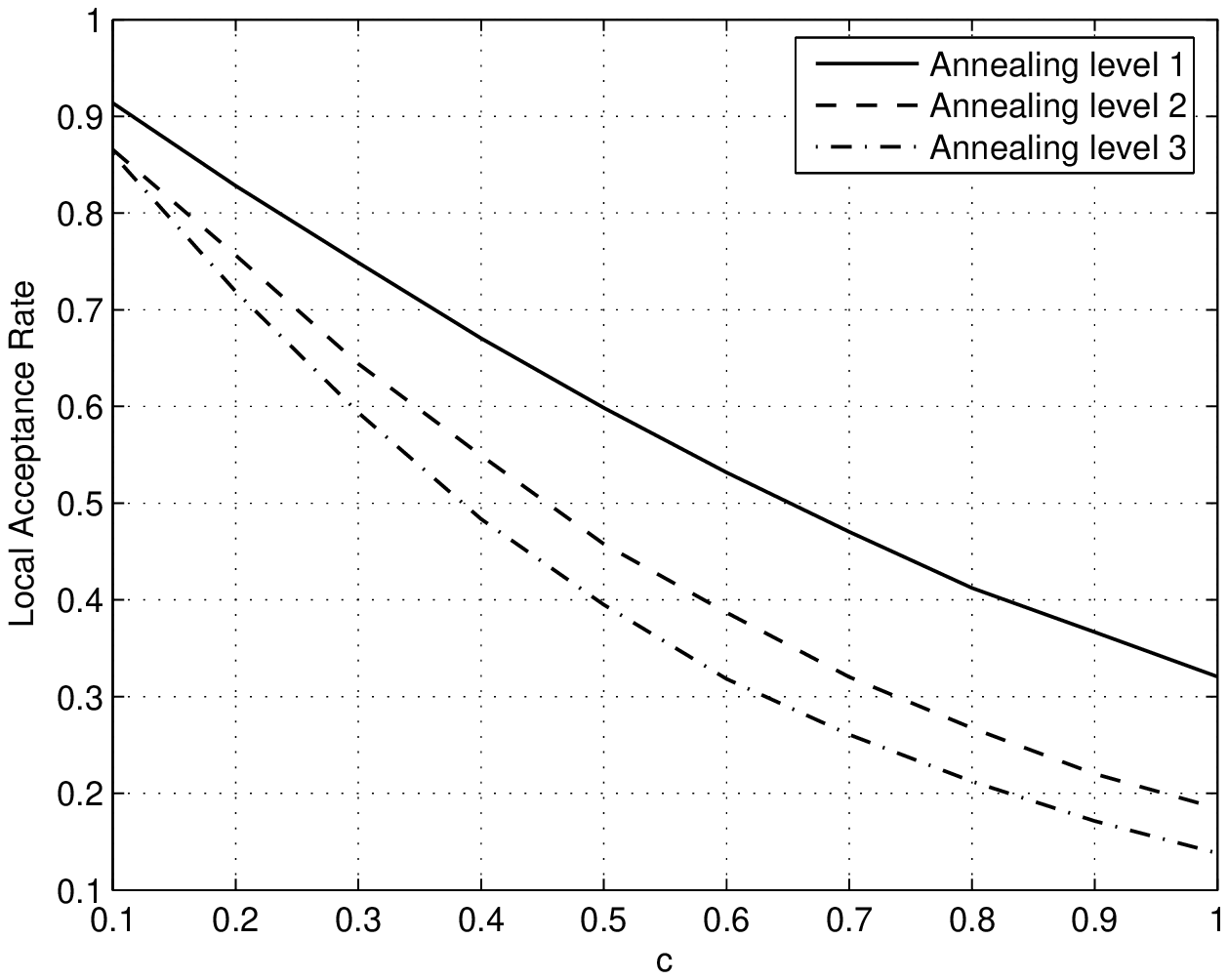}
\caption{\footnotesize Coefficient of variation $\delta$ of the AIMS estimate (top panel), global acceptance rate (middle panel), and
local acceptance rate (bottom panel) as  functions of $c$ for Case 2 ($d=4$) [Example \ref{multGauss}].}\label{fig2}
\end{figure}

\begin{figure}\centering
\includegraphics[angle=0,scale=0.7]{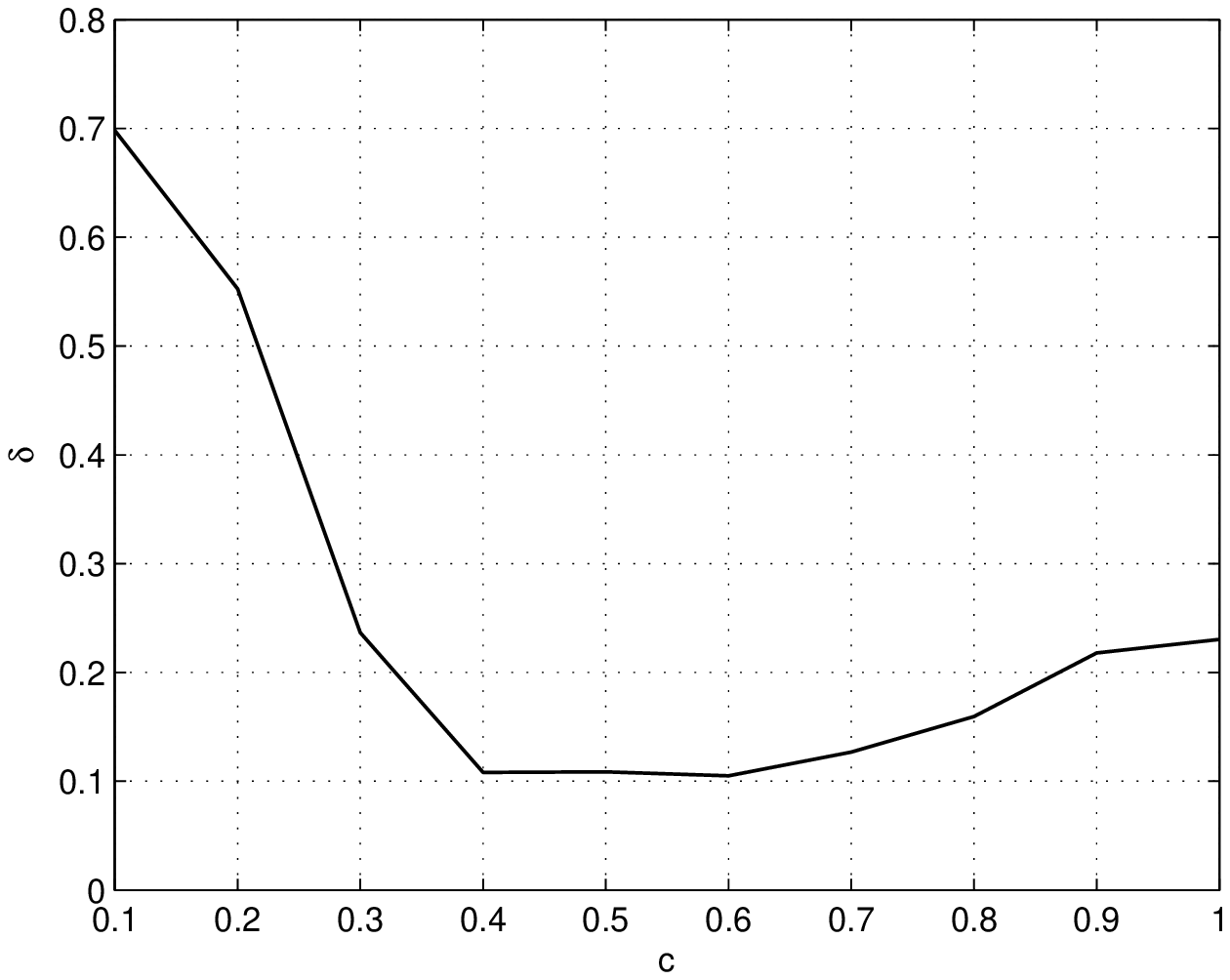}
\includegraphics[angle=0,scale=.7]{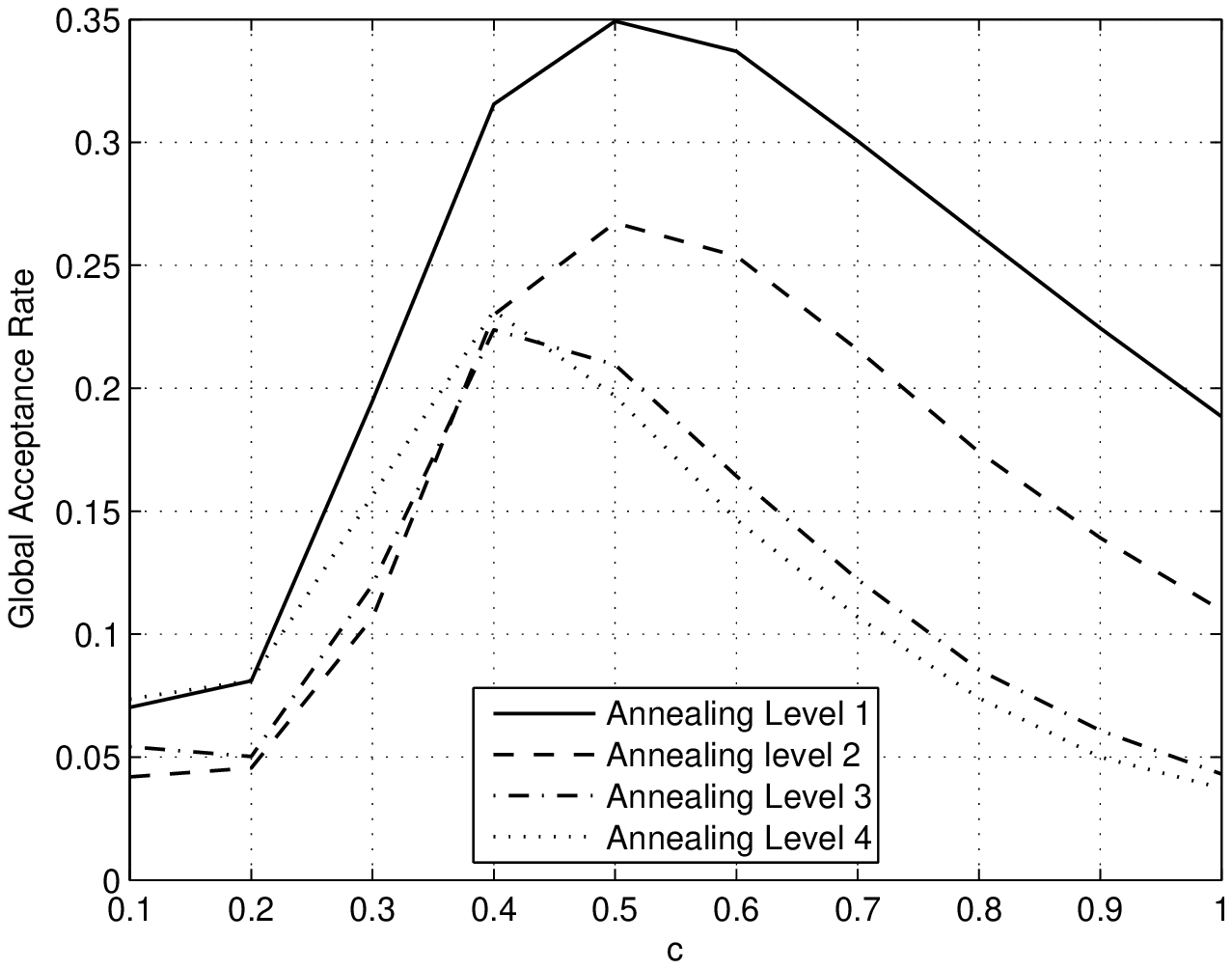}
\includegraphics[angle=0,scale=.7]{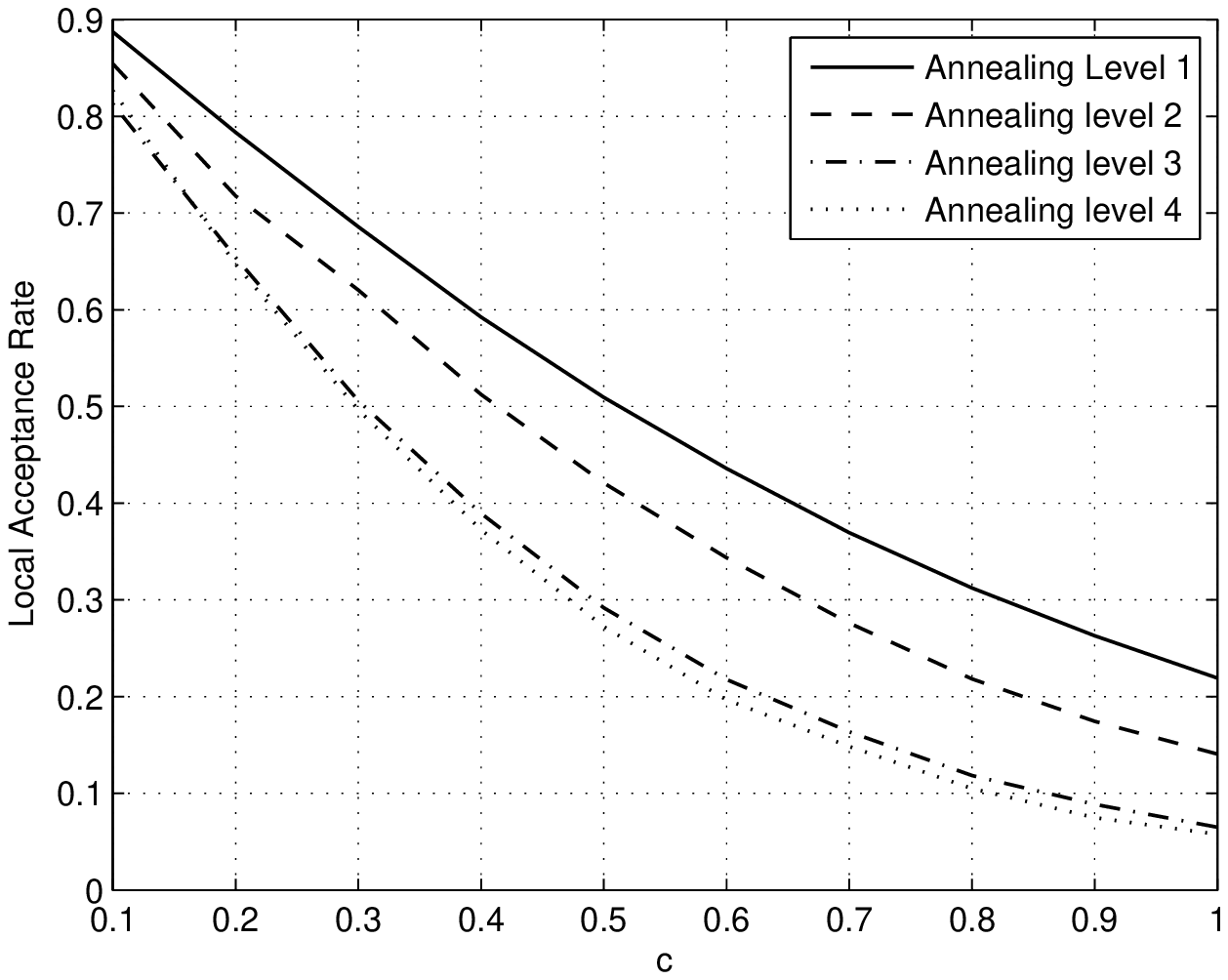}
\caption{\footnotesize Coefficient of variation $\delta$ of the AIMS estimate (top panel), global acceptance rate (middle panel), and
local acceptance rate (bottom panel) as  functions of $c$ for Case 3 ($d=6$) [Example \ref{multGauss}].}\label{fig3}
\end{figure}

\begin{figure}[h]\centerline{
\includegraphics[angle=0,scale=0.8]{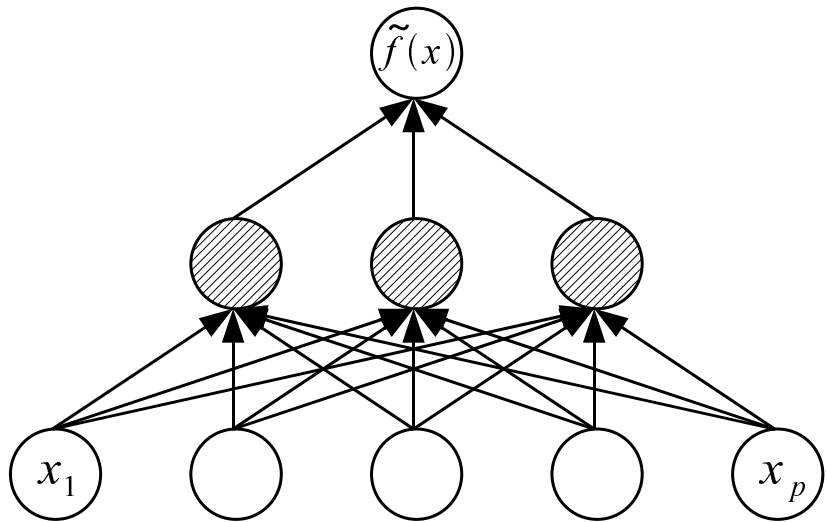}}
\caption{\footnotesize The feed-forward neural network model with one hidden layer (shown by hatching) [Example~4.3].}
\label{FFNNmodel}
\end{figure}

\begin{figure}[h]\centerline{
\includegraphics[angle=0,scale=1]{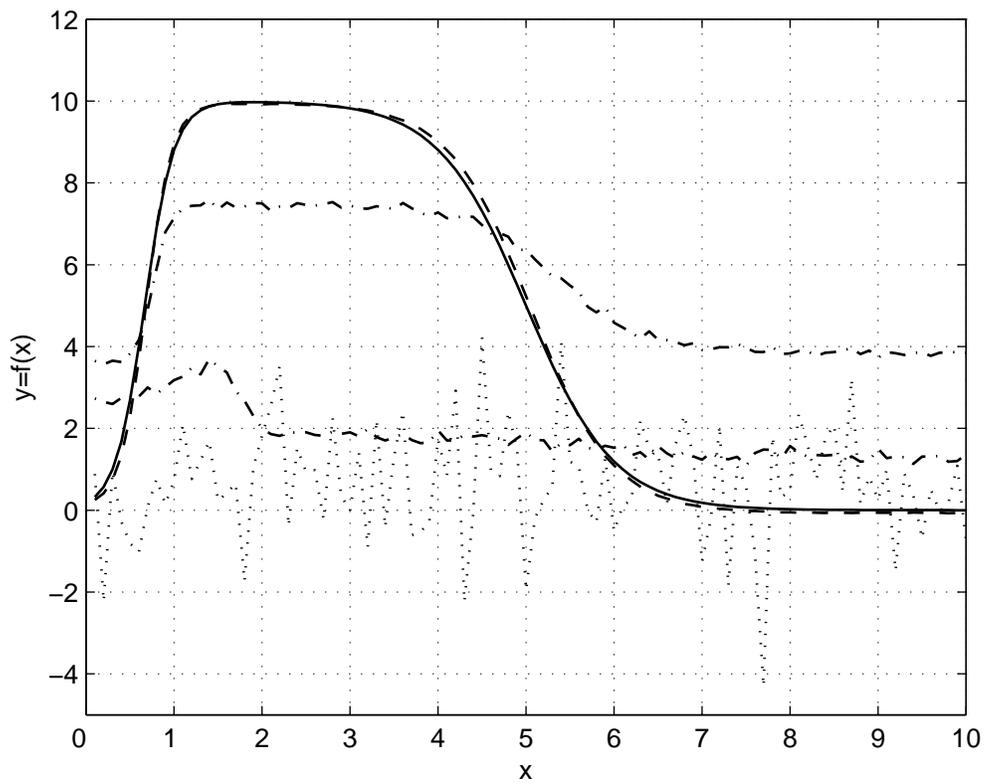}}
\caption{\footnotesize The true function $f(x)$ (solid curve), its posterior approximation $\bar{\hat{f}}_{10}(x)$ (dashed curve) which is
constructed using AIMS,
and ``intermediate annealing approximations'': $\bar{\hat{f}}_0(x)$ (dotted curve) which is based on prior samples, $\bar{\hat{f}}_2(x)$ and
$\bar{\hat{f}}_3(x)$ (dashed-dotted curves) [Example 4.3].}
\label{FFNNexample}
\end{figure}

\begin{figure}[h]\centerline{
\includegraphics[angle=0,scale=1]{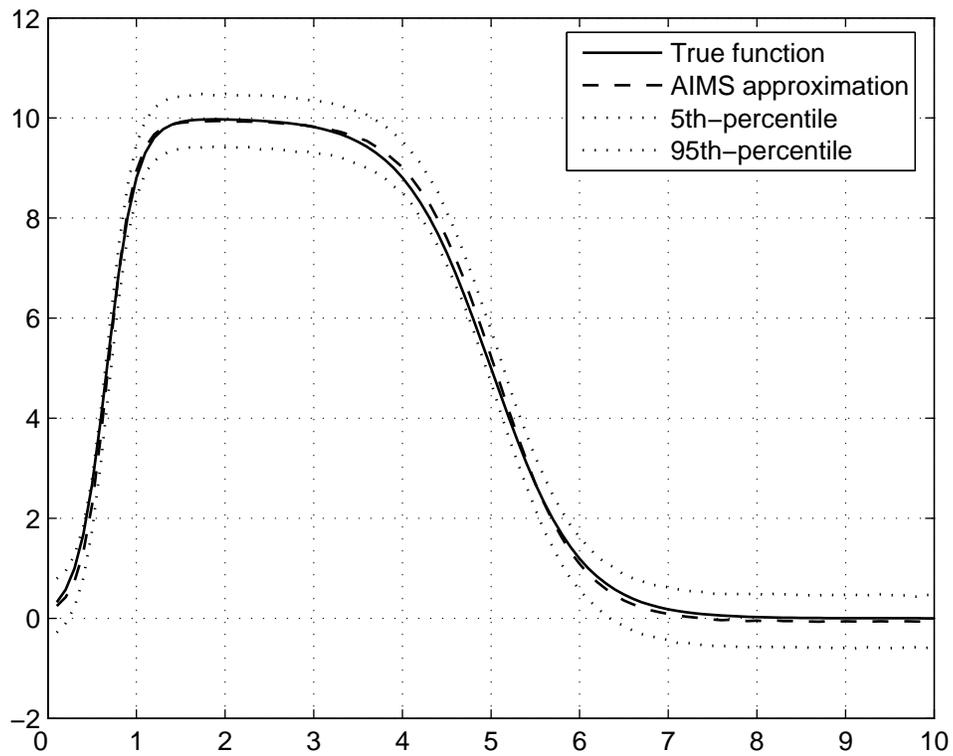}}
\caption{\footnotesize The true function $f(x)$ (solid curve), its AIMS approximation $\bar{\hat{f}}_{10}(x)$ (dashed curve),
and 5th and 95th percentiles
of $\bar{\hat{f}}_{10}(x)$ (dotted curves) [Example 4.3].}
\label{FFNNexample2}
\end{figure}

\end{document}